\documentclass{article}
\usepackage[utf8]{inputenc}
\usepackage{dsfont}
\usepackage{amsmath,amssymb,amsthm,mathtools}
\usepackage{xcolor}
\usepackage[onehalfspacing]{setspace}
\usepackage[nohead]{geometry}
\usepackage{latexsym}
\usepackage{array}
\usepackage[round,comma]{natbib}
\usepackage{pdfsync}
\usepackage{tikz}
\usetikzlibrary{shapes.geometric, arrows.meta, positioning}
\usetikzlibrary{arrows.meta, positioning}

\usetikzlibrary{calc}

\definecolor{uniblau}{cmyk}{1,.9,0,0}
\definecolor{blau}{rgb}{0.2,0.2,0.7}
\definecolor{rot}{RGB}{214,39,40}

\usepackage{hyperref} 
\hypersetup{ 
     backref=true,   
     pagebackref=true,
     citecolor=uniblau,
     hyperindex=true, 
     colorlinks=true, 
     breaklinks=true, 
     urlcolor= uniblau,
     linkcolor= uniblau,
     bookmarks=true, 
     bookmarksopen=true,
     pdftitle={IA — Systemic Risk},
     pdfauthor={Olivier Bos, Stefano Bosi},  
     pdfsubject={}      
}
\usepackage{cleveref}
\usepackage[multiple]{footmisc}

\newtheorem{lemma}{Lemma}

\newtheorem{prop}{Proposition}
\Crefname{prop}{Proposition}
{Propositions}
\newtheorem{corollary}{Corollary}
\newtheorem{claim}{Claim}
\newtheorem{definition}{Definition}

\theoremstyle{definition}

\newtheorem{rem}{Remark}
\newtheorem{ass}{Assumption}

\definecolor{beamer@blendedblue}{rgb}{0.2,0.2,0.7}

\definecolor{beamer@blendedred}{RGB}{214,39,40}

\title{AI Contagion in Social Networks\footnote{We thank Ugo Bolletta, Thai Ha-Huy, Xinyang Wang and participants in the workshop Dialogue between Social Sciences and Humanities (SSH) and Computer Science and Mathematics, Paris-London Economic Theory Workshop at King's College London, Workshop of Central European Program in Economic Theory, Workshop on Microeconomic Theory at Wuhan University, and seminar at Wuhan University of Technology for comments and helpful discussions. Any errors which may remain are ours.}}
\author{{\sc Olivier Bos}\footnote{Universit\'e Paris-Saclay, ENS Paris-Saclay, Centre for Economics at Paris-Saclay. E-mail: olivier.bos@ens-paris-saclay.fr.}\; and {\sc Stefano Bosi}\footnote{Universit\'e Paris-Saclay, Université d'Evry Paris-Saclay, Centre for Economics at Paris-Saclay, EPEE. E-mail: stefano.bosi@univ-evry.fr.}}

\date{July 2026}

\begin{document}

\bibliographystyle{agsm}\addtocontents{toc}{\protect\thispagestyle{plain}}
\maketitle

\begin{abstract}
We study how artificial intelligence (AI) interacts with social communication networks to shape the stability of collective knowledge. Agents exchange information through a network while AI systems generate content and retrain on the aggregate informational environment they influence. This interaction creates a recursive feedback loop in which informational distortions diffuse through society and subsequently feed back into future AI outputs. Despite the high dimensionality of the environment, we show that the long-run dynamics admit a two-dimensional representation whose spectral radius completely characterizes the stability of AI-mediated information systems. We derive a sharp regulatory frontier identifying the minimum filtering required for stability and show how homophily and core-periphery network structures shape systemic informational risk.


\end{abstract}

\newpage
\section{Introduction}

Artificial intelligence (AI) systems increasingly mediate how information is produced, disseminated, and consumed in society. Large language models, recommender systems, search engines, and automated decision tools shape what individuals read, believe, and share. At the same time, these systems are retrained on data generated by the very social environments they influence: search queries, online discussions, social-media interactions, and user feedback all enter AI training pipelines. As a result, AI is no longer a passive technology that merely transmits information. It has become an active component of the social information environment itself. Recent empirical evidence suggests that AI is already reshaping the production of collective knowledge by reducing participation in online knowledge communities \citep{BurtchLeeChen}.

This transformation raises a fundamental concern: what systemic informational risks does society face as AI systems increasingly shape, transmit, and retrain on social information? When AI systems both influence social beliefs and learn from them, informational distortions may propagate through society, accumulate over time, and feed back into future AI outputs.

Recent work emphasizes that many AI-related harms are inherently systemic and cannot be addressed solely through technical fixes or firm-level incentives \citep[see][]{Acemoglu2021}. Similarly, \cite{JacksonEtAl2025} argue for a behavioral science of AI that studies the co-evolution of algorithms and human behavior rather than treating AI as an exogenous tool. Yet the conditions under which feedback between AI retraining and social communication networks generates stable or unstable informational dynamics remain largely unexplored.\footnote{A growing empirical literature documents how AI affects individual behavior and social outcomes, including political reasoning \citep{Bail}, human decision-making \citep{Goergen}, predictability and alignment \citep{ZhouEtAl2025}, and explainability \citep{Miller2019}. These contributions focus on individual or local effects, while our analysis studies aggregate dynamics and systemic stability.} To our knowledge, this paper is the first to develop a formal stability theory of AI-mediated information systems in which AI both shapes social information flows and endogenously retrains on the aggregate informational environment it influences.

This paper studies how the interaction between AI systems and social communication networks affects the evolution, quality, and stability of collective knowledge. As AI-generated content becomes more prevalent, society becomes tightly coupled to algorithmic information production. Even small informational distortions, arising from approximation error, miscalibration, or biased training data, may diffuse through social networks and shape aggregate beliefs. When these socially amplified distortions are subsequently incorporated into AI retraining, the resulting feedback can endogenously destabilize information ecosystems. Understanding when such feedback leads to absorption of errors and when it generates self-reinforcing informational collapse is essential for both economic analysis and the design of effective AI regulation.

To address this question, we develop a tractable dynamic model of the AI–society information system. Society consists of agents connected through a social communication network, while an AI system injects probabilistic information into this network and retrains on the aggregate informational environment it helps create. Informational distortion is measured as divergence from ground truth, and distortions originating from AI enter society alongside human-to-human communication. Two tightly linked feedback forces drive the dynamics. First, an AI contagion syndrome: AI-generated distortions diffuse through the social network and may persist or amplify depending on network structure and verification. Second, an AI social distortion multiplier: AI systems retrain on social data that increasingly reflects their own past outputs, so that circulating distortions feed back into future AI outputs.

Despite the high dimensionality of the underlying system (with heterogeneous agents integrated into a network), we show that the long-run behavior of the AI--society system admits an asymptotic two-dimensional representation in average social distortion and AI distortion. This dimensional reduction arises from the Perron–Frobenius structure of social networks and from the fact that AI retraining responds to aggregate rather than individual signals. The stability of collective knowledge is therefore determined by a single spectral condition. When the feedback loop between social diffusion and AI retraining is sufficiently strong, small informational errors are endogenously amplified over time, leading to systemic informational risk: a regime in which collective knowledge deteriorates even in the absence of malicious actors or irrational behavior.

This reduction yields several central results with direct policy relevance. First, we identify a sharp stability condition for collective knowledge: informational distortions converge if and only if a single amplification factor is below one. When this condition fails, distortions grow endogenously over time, even if initial AI errors are small. Second, we characterize a regulatory frontier that separates stable informational regimes from self-reinforcing misinformation cascades. This frontier identifies the minimum intensity of upstream filtering of AI-generated content required to prevent informational collapse.
Third, we show that network topology shapes systemic informational risk. Homophily increases the persistence of distortions by slowing the decay of subdominant network modes, while core-periphery structures concentrate amplification on a small set of highly influential nodes. Fourth, we provide a simple microfoundation of verification behavior: agents choose costly verification effort in a network game with strategic complementarities, linking individual incentives to macro-level stability and regulatory needs.

\medskip

Our analysis builds on the literature on social learning and networked information aggregation. Beginning with \cite{DeGroot, Bala, DeMarzo, Golub2010} and \cite{acemoglu2011}, this literature studies how dispersed information is aggregated through repeated social interaction and how network structure shapes the evolution and long-run aggregation of beliefs. More recent contributions show that homophily and segregation create learning bottlenecks and persistent disagreement even under rational updating \citep{GolubJackson,AkbarpourMalladiSaberi,Sadler}. A complementary strand of the literature studies conditions under which social learning fails. Persistent mislearning may arise because of excessive influence, network structure, or misspecified updating rules \citep{Acemoglu2010, Bohren, Dasaratha}. These contributions show that collective learning need not aggregate dispersed information efficiently even in the absence of malicious behavior. Our mechanism is different. Rather than arising from stubbornness or incorrect inference, systemic informational risk emerges because AI-generated distortions diffuse through social networks and subsequently become inputs into future AI retraining.

Recent work has begun to incorporate AI into models of collective learning. \cite{Acemoglu2026} study how agentic AI can crowd out human learning effort and thereby weaken the accumulation of collective knowledge. \cite{Acemoglu2026b} introduce an AI aggregator into a DeGroot network and analyze how feedback between AI outputs and population beliefs reshapes long-run information aggregation. Our focus is different but complementary. We hold individual learning behavior fixed and study the dynamic stability of the coupled AI-society system, characterizing when feedback between social diffusion and AI retraining turns informational distortions into an endogenously amplified systemic risk.

Our framework also connects to the literature on systemic risk in financial networks. \cite{ElliottGolubJackson} and \cite{acemoglu} show how network structure determines whether shocks are absorbed or amplified, and how minimal interventions restore stability. We adapt this logic to informational environments. Instead of defaults and balance sheets, the object that propagates is informational distortion; instead of leverage, the amplification channel is AI retraining on socially generated data. The resulting stability condition and policy frontier are informational analogues of macroprudential thresholds.

\medskip

From a policy perspective, our results speak directly to the logic underlying recent regulatory initiatives such as the European Union’s AI Act, the Digital Services Act, and the U.S. NIST AI Risk Management Framework. These frameworks emphasize risk-based, upstream obligations on AI systems and large digital platforms, motivated by concerns about systemic rather than individual harms. Our analysis provides a formal foundation for this approach. We show that systemic informational risk does not arise solely from isolated errors or malicious manipulation, but from structural feedback between AI systems and social communication networks. When AI-generated content diffuses through society and is subsequently absorbed into future AI training, informational distortions can become self-reinforcing. Effective governance must therefore act upstream, scale with systemic exposure, and target the most central nodes in the information network.

The paper is organized as follow.  Section \ref{model} introduces the model of AI-society interaction, describing the social communication network, the diffusion of informational distortions, AI penetration and retraining, and the regulatory filtering mechanism based on information projection. Section \ref{2dim} establishes a dimensional-reduction result, showing that the long-run behavior of the AI-society system can be characterized by a two-dimensional linear system in average social distortion and AI distortion, with stability determined by a single spectral condition. Section \ref{MainResults} analyzes dynamic stability, contagion, and regulation. It derives a sharp policy frontier identifying the minimum filtering intensity required to prevent informational collapse and studies comparative statics with respect to AI penetration, network amplification, and verification. Section \ref{homophily} examines how network topology shapes systemic informational risk through persistence, directional fragility, and finite-horizon stabilization requirements, focusing on the roles of homophily and core-periphery structures. Section \ref{sec:micro} provides a microfoundation of verification behavior through a network game with strategic complementarities, linking individual incentives to macro-level stability and regulatory needs. The concluding remaks section, Section \ref{conclusion}, discusses implications for AI governance, relating the results to risk-based, upstream regulatory frameworks such as the EU AI Act and the Digital Services Act. Proofs are collected in Appendix, in Section \ref{appendix}, as well as the extension about finite-horizon stability and targeted regulation in Section \ref{finitehorizon}. Finally a last Appendix, Section \ref{app:Iproj}, provides explanation for the reader about how the information-theoretic projection (I-projection) is used to model regulatory filtering of distortions originating from AI.

\section{Model: Network, Contagion, and Multiplier}\label{model}

We propose a model of the AI-society information system that captures the two forces highlighted in the introduction: the diffusion of AI-generated distortions through social communication networks and the feedback from aggregate social information into AI retraining. The objective is to isolate the minimal structure required to study systemic informational stability. We represent informational distortion in distributional terms, allow AI systems to inject content into a social network, and let retraining depend on the aggregate informational environment. These ingredients generate a coupled system of diffusion and retraining dynamics whose stability properties admit a tractable and transparent characterization.

\paragraph{The AI Social–Informational Network.} 

Agents interact through a \emph{social–informational network} that specifies who observes or is influenced by whom.\footnote{For intuition, this structure can be visualized as a directed network $G=(V,E)$: each node represents an agent, and an edge $(i,j)$ indicates that agent $i$ places some weight on the information coming from agent $j$. See \cite{Jackson} for overview surveys of social networks.} 

We represent the \emph{pattern of attention}, how each agent allocates attention across informational sources, through a row-stochastic matrix $S=(s_{ij})$, where
$$
s_{ij}\ge 0, \qquad \sum_{j} s_{ij}=1.
$$

The entry $s_{ij}$ measures the share of attention that agent $i$ allocates to information from agent $j$. The matrix $S$ therefore captures the direction and relative weights of informational links, but not the overall intensity with which information is amplified as it circulates in society. $S$ is supposed to be primitive (irreducible and aperiodic).\footnote{Primitivity rules out information cycles where belief patterns would oscillate with periods, and segregated clusters that never communicate. Technically it guarantees unique dominant mode for belief propagation. This is requested to reduce the $n$ to two-dimensional system in Section \ref{2dim}. This is for example used in \cite{GolubJackson}, and usual in the social network literature.}

\medskip

Social interactions, however, may not only mix information but also amplify it. To capture the overall intensity with which information circulates, through repetition, resharing, ranking algorithms, or platform virality, we introduce a scalar amplification parameter $\alpha>0$. We define the effective social propagation operator as
$$
W := \alpha S.
$$

While $S$ is row-stochastic and therefore has spectral radius equal to one, the matrix $W$ is not row-stochastic in general.\footnote{Since $S$ is primitive and row-stochastic, the associated Markov chain is ergodic.  In particular, it admits a unique stationary distribution $\pi$ satisfying $\pi^\top S=\pi^\top$, and $S^t \to \mathbf{1}\pi^\top \quad \text{as } t\to\infty$. Moreover, by the Perron-Frobenius theorem, $1$ is a simple eigenvalue of $S$  and all other eigenvalues have modulus strictly smaller than one. Hence there exists a constant $C>0$ such that $\|S^t-\mathbf{1}\pi^\top\|\le C\,|\lambda_2(S)|^t$, where $\lambda_2(S)$ denotes the second-largest eigenvalue of $S$ in modulus. Since $W=\alpha S$, it follows that
$W^t=\alpha^t S^t=\alpha^t\mathbf{1}\pi^\top +O\!\left((\alpha|\lambda_2(S)|)^t\right)$. The limit matrix $\mathbf{1}\pi^\top$ is rank one and maps any vector $x$ to $\mathbf{1}(\pi^\top x)$. Thus distortions asymptotically align with the stationary influence weights $\pi$, while their magnitude evolves at rate $\alpha^t$.} Its spectral radius is given by
$$
\rho(W)=\alpha,
$$

\noindent which we interpret as the \emph{systemic amplification index} (or \emph{network amplification factor}). Higher values of $\alpha$ correspond to environments in which informational distortions are more strongly magnified as they propagate through the social network due to denser connectivity, stronger engagement, or algorithmic amplification, and thus to a greater potential for systemic informational contagion.

\medskip

Each agent $i$ receives a fraction $\mu_i\in[0,1]$ of her informational input from AI systems (search engines, large language models, recommender algorithms, decision-support tools). The parameter $\mu_i$ thus describes how intensively agent $i$ relies on AI-generated content when forming beliefs. In a low-$\mu_i$ environment, most of her information still comes from human or traditional sources; in a high-$\mu_i$ environment, AI acts as the primary gatekeeper of knowledge. We define the \emph{AI penetration rate} as the population average
$$
\mu = \frac{1}{n}\sum_{i=1}^n \mu_i.
$$

\noindent which measures how deeply AI is integrated into the informational fabric of society as a whole. A higher $\mu$ corresponds to a world in which AI systems mediate a larger share of news consumption, search queries, recommendations, and decision support. Economically, $\mu$ plays the role of an informational leverage parameter: as $\mu$ increases, errors or biases originating in AI affect more agents simultaneously, and the potential for system-wide amplification of distortions grows. Given $\mu$, we define the AI-adjusted weighted influence matrix

$$
\tilde W = (1-\mu)W = (1-\mu)\alpha S,
$$
\noindent so that $\rho(\tilde W)=\alpha(1-\mu)$ is the effective amplification index once AI penetration is taken into account. The larger is $\mu$, the weaker is the relative role of human-to-human communication in shaping beliefs, and the more the dynamics of distortions are driven by AI origin content.\footnote{At the individual level, agents may differ in their reliance on AI-generated content, $\mu_i$. Aggregating distortions yields an AI-injection term that depends on the population average $\mu=\frac{1}{n}\sum_i\mu_i$. Replacing $(1-\mu_i)$ by $(1-\mu)$ corresponds to a standard mean-field aggregation appropriate for our focus on aggregate stability. Since $\tilde W=(1-\mu)W$, its spectral radius satisfies $\rho(\tilde W)=(1-\mu)\rho(W)=(1-\mu)\alpha$, so $\alpha$ captures structural amplification while $(1-\mu)$ scales the share of information transmitted through human-to-human interaction rather than AI sources.}

\paragraph{The AI Contagion.}

We represent each agent's informational state as a probability distribution $P_{i,t}$ over the states of knowledge and measure distortions using the Kullback–Leibler (KL) divergence. This choice is natural for two reasons. First, AI systems generate probabilistic information: their outputs are samples from model distributions, and evaluating such outputs in terms of distance to truth is most naturally done in distributional space. Second, KL divergence is the standard information-theoretic measure of informational distortion. It captures not only level errors but also systematic biases, overconfidence, and distributional drift, all of which are central to AI-induced misinformation.

Each agent $i$ holds at date $t$ a belief distribution $P_{i,t}$ over states of knowledge, while $\hat P$ denotes the ground truth. We measure informational distortion using the Kullback--Leibler (KL) divergence,
$$e_{i,t} := D_{\mathrm{KL}}(P_{i,t}\,\|\,\hat P),$$
which captures not only pointwise errors but also systematic biases, overconfidence, and distributional drift. This choice is natural in the context of modern AI systems, whose outputs are probabilistic and whose errors are inherently distributional rather than binary.

We do not model the full nonlinear evolution of belief distributions. Instead, we work with a reduced-form law of motion for the \emph{magnitude} of informational distortions. When beliefs remain close to the truth, KL divergence admits a local quadratic approximation, so that the expected evolution of distortion magnitudes induced by linear belief exchange can be represented, to first order, by a linear recursion.\footnote{The linear belief-updating rule adopted here is a DeGroot-type non-Bayesian learning rule. \citet{Molavi} provide behavioral foundations for such updating under imperfect recall, whereby agents treat their neighbors' current beliefs as sufficient statistics for the information histories that generated them. Our linear recursion concerns informational distortion rather than beliefs themselves and is instead obtained as a local reduced-form approximation, as explained in Appendix \ref{app:derivation_contagion}.}

Specifically, we assume that informational distortions evolve according to
\begin{equation}
\label{eq_contagion}
e_{t+1} = (1-\delta)\tilde W e_t + (1-\delta)\mu\,\mathbf{1}\,q_t,
\end{equation}
where $e_t=(e_{1,t},\ldots,e_{n,t})^\top$, $\tilde W=(1-\mu)W$ captures human-to-human transmission of distortions through the social network, $\delta\in(0,1)$ represents verification or forgetting, and $q_t$ measures the distortion in AI-generated content (expressed in the same informational units).

Equation \eqref{eq_contagion} should be interpreted as a first-order approximation to the evolution of distortion magnitudes near the truth, rather than as a literal updating rule for belief distributions. It combines three forces in a parsimonious way: social propagation of informational error, attenuation through verification, and exogenous injection of distortions originating from AI. An appendix, Section ~\ref{app:derivation_contagion}, provides a derivation sketch based on linear mixing of belief distributions and a local quadratic expansion of the KL divergence.

We say that an \emph{AI contagion syndrome} prevails if, holding the AI distortion process fixed, social transmission amplifies informational distortions rather than attenuating them. This notion describes situations in which networked interaction overwhelms verification and correction, so that informational errors spread endogenously through society.\footnote{Proposition~\ref{prop1} formalizes this idea and characterizes the precise conditions under which such contagion arises.}

\medskip

\paragraph{The AI Multiplier.}
The AI multiplier formalizes the feedback loop from society to AI. Suppose that at each period, new AI models are trained on a dataset combining a share $\gamma$ of initial/previous data and $(1-\gamma)$ of social content, possibly contaminated by misinformation.  
Then AI error evolves as
\begin{equation}
\label{eq_multiplier}
q_{t+1} = \gamma q_t + (1-\gamma)\bar e_t,
\end{equation}
where $\bar e_t=\frac{1}{n}\sum_i e_{i,t}$ is the average societal distortion. We interpret the average distortion $\bar e_t$ as the state of collective knowledge in society, summarizing the informational environment after distortions have propagated through the network. Lower $\gamma$ implies stronger self-reinforcement and faster degradation of knowledge accuracy. This is the \emph{AI social distortion multiplier}: AI systems retrain on the average social distortion $\bar e_t$.

It captures the fact that AI systems retrain on data that increasingly contains their own past outputs. At early stages, when both $q_t$ and $\bar e_t$ are small, the AI mostly sees clean data, so quality evolves slowly, improving or drifting only modestly. Once distortions begin to spread through the social network, however, a growing share of the AI’s training corpus becomes contaminated. The more distorted the environment becomes, the more strongly these distortions feed back into AI training. This generates an economically concave pattern of AI quality: improvements slow down as the system gets better, but degradations accelerate quickly once errors enter the feedback loop. In this sense, the multiplier describes how self-generated distortions accumulate and can eventually overturn initial gains in accuracy.\footnote{Repeated substitution of the multiplier equation leads to the explicit form $q_{t+k} = \gamma^k q_t + \sum_{s=0}^{k-1}\gamma^s (1-\gamma)\bar e_{t+k-1-s}$. If $\bar e_t$ remains small, the weighted sum in the second term remains small, so the path of $q_t$ changes slowly. When $\bar e_t$ begins to rise, the geometric rate $\gamma^s$ places more weight on recent data, amplifying the effect of rising distortions. The cumulative contribution of these distorted inputs therefore accelerates even though the recurrence is linear. This generates a concave deterioration of AI quality over time: minimal changes when distortions are rare, and rapid deterioration once distortions accumulate.}

\paragraph{Regulatory filtering via I-projection.}

Before distorted AI output is allowed to diffuse through the social network, the regulator filters a
fraction $\kappa\in[0,1]$ of the AI’s error using an information-theoretic projection. The filter
operates by mapping an AI-generated distribution $Q$ onto the closest truth-consistent distribution
in a feasible set $\mathcal{\hat P}$. Formally, the regulator applies the I-projection \citep{Csiszar}:
$$
P^\star
    = \arg\min_{P\in\mathcal{\hat P}} D_{\mathrm{KL}}(P\|Q).
$$
This operator minimizes the Kullback–Leibler divergence between $Q$ and the set of admissible distributions, ensuring that the correction is the smallest informational change required to restore truth-compatibility (the states of knowledge). Economically, $\kappa$ indexes the strength of filtering: $\kappa=0$ corresponds to laissez-faire, while $\kappa=1$ achieves perfect alignment with the truth distribution $\hat P$.

Using the I-projection rather than an ad hoc penalty is crucial. KL divergence measures distortion, but it does not by itself prescribe how to reduce it. The I-projection supplies an explicit, optimization-based rule that removes errors originating from AI in the most information-preserving way. In effect, the regulator insulates society from a portion of the AI's error before this error enters the network and becomes subject to contagion.

\medskip 

In our analysis, this filtering reduces the effective AI error term in the contagion equation from $q_t$ to $(1-\kappa)q_t$, thereby directly influencing the injection parameter $b=(1-\delta)\mu(1-\kappa)$ in the two-dimensional dynamic (see Section \ref{2dim}). The I-projection thus provides a structural and economically interpretable tool for upstream regulation of AI-generated information.

\medskip

The parameter $\kappa$ captures the intensity of upstream filtering applied to AI-generated information before it diffuses through society. In practice, this corresponds to a wide range of realistic regulatory mechanisms: mandatory accuracy evaluations, fact-checking layers, safety alignment constraints, content moderation obligations, adversarial testing, dataset auditing, and other quality-control processes included in the EU AI Act, the Digital Services Act (DSA), and the U.S. NIST AI Risk Management Framework. 

\paragraph{Regulatory filtering and effective distortion.} At a conceptual level, AI systems produce distorted probability distributions over the states of knowledge, and the regulator seeks to correct these distortions in a minimally invasive way. The I-projection provides exactly such an operator: it maps an AI-generated distribution onto the closest truth-consistent distribution by minimizing the Kullback—Leibler divergence. In this sense, the I-projection implements the smallest information-theoretic correction required to restore truth-compatibility.

In the dynamic analysis, however, the model does not track the full probability distributions themselves, but only the magnitude of distortion they induce. The regulatory parameter $\kappa$ therefore measures the fraction of distortion originating from AI removed upstream before diffusion through the social network. This is why the contagion equation involves the term $(1-\kappa)q_t$: for the linear dynamics of diffusion and feedback, only the total amount of distortion that remains after filtering affects the dynamics, not the fine details of how the distribution is reshaped.

Formally, once the I-projection has been applied, the remaining distortion enters the system as a scaled shock. The subsequent results—dimensional reduction, stability conditions, and policy frontiers—depend on this residual distortion through $(1-\kappa)$, while the information-theoretic properties of the projection justify why $\kappa$ has a clear and non-arbitrary interpretation. Thus, KL divergence and I-projection provide the microfoundation for regulatory filtering, even though the aggregate dynamics can be expressed in terms of a single scalar control.

\section{Aggregation and Dimensional Reduction of the AI-society System}\label{2dim}

The diffusion of misinformation evolves in a $(n+1)$-dimensional space: $e_t \in \mathbb{R}^n$ represents the distortion held by each of $n$ agents, while $q_t \in \mathbb{R}$ captures the AI's intrinsic distortion. Taken together, this system is at first glance analytically intractable. Yet the economic forces at play exhibit a strong regularity. Social learning rapidly reduces heterogeneity in local distortions, while AI retraining aggregates information at a population-wide scale. As a consequence, the long-run interaction between AI and society admits a two-dimensional aggregate representation in the stable regime: the average distortion in society $\bar e_t$ and the AI distortion $q_t$.

\medskip

This dimensional reduction is economically natural. When interactions are frequent and influence diffuses across the network, local misinformation tends to become homogenized across agents. In social-information networks under linear averaging \citep{DeGroot}, repeated exchanges drive beliefs toward a common value across agents: beliefs converge to a consensus vector proportional to $\mathbf{1}$, with influence weights determined by the stationary distribution of the attention matrix. As a result, the state of the network becomes well summarized by its aggregate component. Deviations from this consensus decay geometrically relative to the dominant eigenmode at rate $|\lambda_2(S)|$, so that in the long run only this mode remains. Meanwhile, the AI system does not condition on individual distortions but on the overall quality of the data it collects from society. The feedback channel between AI and society is therefore mediated by the average distortion $\bar e_t$ rather than by the full vector $e_t$, since the AI system learns from the aggregate data generated by society. As a result, only the aggregate component of distortions feeds back into the dynamics.

\medskip

To analyze this interaction, it is convenient to collect the two aggregate variables into the vector
$$
x_t := \begin{bmatrix} \bar e_t \\ q_t \end{bmatrix},
\qquad 
\bar e_t := \tfrac{1}{n}\mathbf{1}^\top e_t.
$$
The first component $\bar e_t$ summarizes the informational distortion circulating in society, while the
second component $q_t$ represents the AI system's internal distortion.

This reduction is not an assumption but follows from the spectral properties of the network. As discussed above, the AI responds only to the average quality of social information, and the network aligns distortions along its dominant eigenmode, with subdominant modes vanishing geometrically relative to the dominant component. This suggests that the long-run behavior of the $(n+1)$-dimensional system admits a two-dimensional representation in terms of the state $x_t$ in the stable regime.

Substituting the contagion equation \eqref{eq_contagion} and the feedback equation \eqref{eq_multiplier} leads
to the linear recursion
\begin{equation}\label{eq_2dim}
x_{t+1} = A(\mu,\kappa)\,x_t,
\qquad
A(\mu,\kappa) :=
\begin{bmatrix}
m & b \\
1-\gamma & \gamma
\end{bmatrix},
\end{equation}
where
$$
m := (1-\delta)(1-\mu)\,\alpha
\quad\text{and}\quad
b := (1-\delta)\mu(1-\kappa).
$$

As $S$ is primitive and row-stochastic, Perron--Frobenius implies that the right eigenvector associated with the dominant eigenvalue is $\mathbf{1}$ and the corresponding left eigenvector is the stationary distribution $\pi$, normalized so that $\pi^\top \mathbf{1}=1$. Since $W=\alpha S$, the two matrices share the same eigenvectors. Under this normalization, the projection of the AI-injection term onto the dominant eigenspace introduces no additional constant. Consequently, the aggregate coefficients in the reduced two-dimensional system are exactly $m$ and $b$ as defined above, and no rescaling of $q_t$ is required. The parameter $m$ captures the strength of internal social diffusion, combining the survival rate of distortions $(1-\delta)$, the reliance on human (non-AI) communication $(1-\mu)$, and the network amplification index $\alpha$. The parameter $b$ measures the effective injection of distortion originating from AI into society, taking into account verification $(1-\delta)$, AI penetration $\mu$, and regulatory filtering $(1-\kappa)$.

\medskip

Although the system \eqref{eq_2dim} has been motivated heuristically, the next lemma formalizes the dimensional-reduction argument. Let $\lambda_2(S)$ denote the eigenvalue of the primitive row-stochastic matrix $S$ with second-largest absolute value. Since $S$ is primitive and row-stochastic, $|\lambda_2(S)|<1$.\footnote{For primitive row-stochastic matrices, $\lambda_2(S)$ determines the rate at which powers of $S$ converge to the rank-one Perron projection. See, e.g., \cite{DeGroot} or \cite{GolubJackson}.} The aggregate dynamics can be written as a two-dimensional recursion driven by $A(\mu,\kappa)$ up to a perturbation term that captures the contribution of subdominant network modes. These modes vanish geometrically relative to the dominant eigenmode, in the sense that their contribution decays at a rate determined by $|\lambda_2(S)|$. Moreover, for any $\rho\in(|\lambda_2(S)|,1)$, the perturbation is bounded by $C(m\rho)^t$ and therefore vanishes in absolute value whenever $m\rho<1$. Thus, in the stable regime, the long-run behavior is determined by the spectral radius of $A(\mu,\kappa)$.

\begin{lemma}[Dimensional reduction and perturbation bounds]\label{lemma1}

Define $m=(1-\delta)(1-\mu)\alpha$ and $b=(1-\delta)\mu(1-\kappa)$. Then:
\begin{enumerate}
\item[(i)] There exists a sequence $(\varepsilon_t)_{t\ge0}$ such that
$$
\bar e_{t+1}=m\bar e_t + b q_t + \varepsilon_{t+1}.
$$
Moreover, for any $\rho\in(|\lambda_2(S)|,1)$ there exists a constant $C<\infty$ such that $|\varepsilon_{t+1}|\le C\,(m\rho)^{t+1}$. In particular, the error term vanishes geometrically relative to the dominant eigenmode in the sense that $|\varepsilon_{t+1}|/m^{t+1}\le C\rho^{t+1}\to 0$. If $m\rho<1$, then $(\varepsilon_t)$ vanishes geometrically in absolute value.
\item[(ii)] Letting $x_t=(\bar e_t,q_t)^\top$, the aggregate dynamics satisfy
$$
x_{t+1}=A(\mu,\kappa)x_t + \eta_{t+1},\;\; \eta_{t+1}=(\varepsilon_{t+1},0)^\top,
$$
with $\|\eta_{t+1}\|\le C\,(m\rho)^{t+1}$ and hence $\|\eta_{t+1}\|/m^{t+1}\le C\,\rho^{t+1}$.

\item[(iii)] If $m\rho<1$, then $\eta_t\to 0$ geometrically and the $(n+1)$-dimensional system is asymptotically equivalent to the two-dimensional system $x_{t+1}=A(\mu,\kappa)x_t$.
\end{enumerate}
\end{lemma}

The pair $(\bar e_t, q_t)$ characterizes the interaction between social diffusion and AI retraining. Because network dynamics concentrate distortions along the dominant eigenmode and the AI system retrains on aggregate information, the long-run feedback between society and AI is mediated by these two sufficient statistics. Lemma~\ref{lemma1} formalizes this reduction by showing that the aggregate evolution of $(\bar e_t,q_t)$ can be written as a two-dimensional linear recursion driven by $A(\mu,\kappa)$, up to a perturbation term that captures the contribution of subdominant network modes. This perturbation vanishes geometrically \emph{relative to the dominant eigenmode}, in the sense that its magnitude relative to the dominant component decays geometrically at a rate $|\lambda_2(S)|$, and it vanishes in absolute value whenever $m\rho<1$ for some $\rho\in(|\lambda_2(S)|,1)$. As a result, whenever distortions do not explode, systemic informational risk and stability are determined in the long run by the matrix $A(\mu,\kappa)$.

\paragraph{Technical intuition.}
Because $W$ is primitive, its powers admit the Perron--Frobenius decomposition
$$
W^t=\alpha^t \mathbf 1 \pi^\top + R_t \text{ with } \|R_t\| \le C_2 (\alpha\rho)^t
$$
for any $\rho\in(|\lambda_2(S)|,1)$ and some constant $C_2>0$. The first term corresponds to the dominant Perron projection, while $R_t$ collects subdominant network modes that decay geometrically relative to the dominant eigenmode, in the sense that their contribution is bounded by $(\alpha\rho)^t$ for any $\rho>|\lambda_2(S)|$. The effective social propagation matrix is $M=(1-\delta)(1-\mu)W=mS$ with $m=(1-\delta)(1-\mu)\alpha$, so $M^t$ inherits the same eigenvectors and has subdominant modes bounded by $(m\rho)^t$. When we aggregate distortions by $\mathbf{1}^\top/n$, the AI-injection term (which is proportional to $\mathbf{1}$) lies exactly in the dominant right-eigenspace of $M$, so it has no projection on subdominant modes. As a consequence, the only remainder term in the aggregate recursion for $\bar e_t$ comes from the initial condition and satisfies $|\varepsilon_t|\le C(m\rho)^t$ or, equivalently, $|\varepsilon_t|/m^t\le C\rho^t$.

Coupling the resulting scalar recursion for $\bar e_t$ with the AI feedback equation for $q_t$ yields a two-dimensional linear system for $x_t=(\bar e_t,q_t)^\top$ with an additive perturbation $\eta_t=(\varepsilon_t,0)^\top$. This perturbation vanishes geometrically relative to the dominant eigenmode, in the sense that $\|\eta_t\|/m^t\le C\rho^t$, and its absolute magnitude is bounded by $C(m\rho)^t$. In particular, it vanishes in absolute value whenever $m\rho<1$. In such regimes, the long-run behavior is asymptotically equivalent to the unperturbed two-dimensional system $x_{t+1}=A(\mu,\kappa)x_t$, and stability is determined by the spectral radius of $A(\mu,\kappa)$.

The perturbation term $\eta_t$ in Lemma~\ref{lemma1} reflects the contribution of subdominant eigenmodes of the social propagation operator. Since $|\lambda_2(S)|<1$, these modes decay geometrically relative to the dominant eigenmode.\footnote{The spectral reduction is usual in a wide range of works. Linear network-learning and information-diffusion models rely on the Perron--Frobenius structure to characterize asymptotic beliefs and convergence speeds \citep{GolubJackson, DeGroot}, while in systemic-risk models, contagion thresholds and amplification factors are determined by the dominant eigenvalue of the interdependence matrix \citep{acemoglu, ElliottGolubJackson, GolubJackson}. Moreover, related aggregation results for non-viral contagion processes appear in \cite{Sadler}.} In our setting, the AI-injection term is proportional to $\mathbf{1}$ and therefore lies exactly in the dominant right-eigenspace of $M$, so it does not generate a projection on subdominant modes. Consequently, $\eta_t$ captures only the residual contribution of subdominant modes from the initial condition, and Lemma~\ref{lemma1} implies the bound $\|\eta_t\|/m^t\le C\rho^t$. Thus, the reduced system captures the dynamics relevant for policy analysis in the long run: whether distortions are absorbed or amplified, and at what rate.

\begin{rem}
Although the full system is $(n+1)$-dimensional, the pair $(\bar e_t, q_t)$ provides a sufficient statistic for its long-run behavior in the stable regime. In particular, when the perturbation vanishes in absolute value (e.g.\ if $m\rho<1$ for some $\rho\in(|\lambda_2(S)|,1)$), the dynamics are asymptotically equivalent to the two-dimensional system $x_{t+1}=A(\mu,\kappa)x_t$. In this sense, stability is characterized by the spectral radius of $A(\mu,\kappa)$, and the regulator influences systemic risk only through the injection parameter $b$.
\end{rem}

\paragraph{Relation to static contagion models and limits of static aggregation.} Recent work shows that, outside of viral regimes, aggregate diffusion patterns in networks can often be characterized by a small number of spectral statistics rather than the full graph. In particular, \citet{Sadler} studies simple contagion processes in large random networks and shows that expected diffusion outcomes depend on low-dimensional summaries such as degree moments and eigenvalues. These results provide a rigorous justification for why high-dimensional network dynamics may admit tractable aggregate representations.

Our setting differs in a fundamental way. The analysis of \cite{Sadler} applies to \emph{static} contagion processes in which a fixed shock propagates through the network and eventually dies out. The source of infection is exogenous and does not respond to the diffusion it induces. As a result, the main question concerns the size and reach of diffusion from a given seed, not the stability of the system over time.

In contrast, the main mechanism in our model is a \emph{dynamic feedback loop} between society and AI. Distortions propagate through the social network, are aggregated into the average distortion $\bar e_t$, and then feed back into the AI system through retraining. This feedback endogenously modifies the future source of distortion itself. The relevant question is therefore not the expected size of a diffusion episode, but whether distortions are \emph{absorbed or amplified across iterations}.

Lemma~\ref{lemma1} shows that, despite this feedback, the high-dimensional AI-society system admits an asymptotically two-dimensional representation in the long run: subdominant network modes vanish geometrically relative to the dominant eigenmode, and in the stable regime the dynamics are asymptotically equivalent to the corresponding two-dimensional system. Our contribution is to characterize the stability of this dynamic system, derive the amplification threshold $\rho(A)=1$, and identify the policy frontier $\kappa^\star(\mu)$ that separates stable informational regimes from systemic informational risk. These objects have no analogue in static contagion frameworks and arise only once feedback between diffusion and generation is explicitly modeled.

\section{Dynamic Stability, Contagion, and Regulation}\label{MainResults}

We first study the environment in which the AI has an exogenous and constant level of distortion $q_t\equiv q>0$. The AI injects a fixed amount of error each period, coming from a structural bias, training flaw, or persistent misalignment, but does not update based on distortions circulating in society. This shuts down the AI–society–AI feedback loop and isolates the propagation effect of the social network alone. Therefore, the contagion dynamics reduce to a linear recursion with constant input

$$
e_{t+1}=Me_t+\underbrace{(1-\delta)\mu(1-\kappa)q\,\mathbf{1}}_{\text{constant injection}} \text{ with } M:=(1-\delta)\tilde W=(1-\delta)(1-\mu)W.
$$

The matrix $M$ summarizes how past distortions persist: $(1-\mu)$ is the fraction of information transmitted through human links rather than AI, and $(1-\delta)$ is the survival rate of
misinformation after verification or forgetting.

\paragraph{Pure-network contagion benchmark.}
Proposition~\ref{prop1} establishes when, in the absence of AI feedback, the social network absorbs the AI constant error or amplifies it without bound. Even with no feedback from AI, instability may arise purely from the network structure: the matrix $M$ can either attenuate or magnify repeated distortion shocks.

\begin{prop}[No AI multiplier/No AI feedback]
\label{prop1}
Suppose $q_t\equiv q>0$ is constant. Then:
\begin{itemize}
\item Bounded distorsions: $e_t$ is bounded if and only if $\rho(M)<1$, equivalently $\alpha(1-\delta)(1-\mu)<1$.
\item Bounded limit: If $\rho(M)<1$, then $e_t\to e^\star=(I-M)^{-1}(1-\delta)\mu(1-\kappa)q\,\mathbf{1}$, at geometric rate $\rho(M)$.
\item Contagion: If $\rho(M)\ge 1$, then $\|e_t\|\to\infty$.
\end{itemize}
\end{prop}

If $\rho(M)<1$, distortions remain under control. Hence, repeated interactions reduce the influence of past errors, and society converges to the finite steady-state level $e^\star$ (long-run level of misinformation). The resolvent $(I-M)^{-1}$
plays the role of a contagion multiplier, as in Keynesian or input–output models. If $\rho(M)\ge1$, distortions are amplified faster than verification can remove them, and even a small constant AI bias generates unbounded misinformation.

\medskip

Proposition~\ref{prop1} shows that instability may arise from network structure alone. We now activate the AI multiplier. The AI learns from the informational environment it helps create:
$$
q_{t+1}=\gamma q_t + (1-\gamma)\bar e_t,
\qquad
\bar e_t = \frac{1}{n}\mathbf{1}^\top e_t.
$$

Together with the contagion equation, the aggregate state satisfies the representation
$$
x_{t+1}=A(\mu,\kappa)x_t+\eta_{t+1} \text{ with } x_t=(\bar e_t,q_t)^\top,
$$

\noindent where the perturbation term $\eta_{t+1}$ captures the contribution of subdominant network modes as characterized in Lemma~\ref{lemma1}. These modes decay geometrically relative to the dominant eigenmode, in the sense that their relative contribution is determined by $|\lambda_2(S)|$. More precisely, for any $\rho\in(|\lambda_2(S)|,1)$, Lemma~\ref{lemma1} implies $\|\eta_{t+1}\|\le C(m\rho)^{t+1}$. As a consequence, long-run systemic amplification is determined by the spectral radius $\rho(A(\mu,\kappa))$ of the two-dimensional feedback matrix. Proposition~\ref{prop2} therefore characterizes the instability threshold of the full $(n+1)$-dimensional system.

\begin{prop}[Systemic contagion \& feedback loop]
\label{prop2}
Let $x_{t+1}=A(\mu,\kappa)x_t$ as in \eqref{eq_2dim}. Then:
\begin{itemize}
\item Absorbed distorsions: $x_t\to 0$ for all initial $x_0$ if and only if $\rho\big(A\big)<1$.
\item Systemic contagion: There exists $x_0\neq 0$ with $\|x_t\|\to\infty$ if and only if $\rho\big(A\big)>1$.
\item The exponential growth/decline rate is $\log \rho\big(A\big)$.
\end{itemize}
The boundary $\rho(A)=1$ satisfies $\det(A-I)=0$, i.e. $b=1-m$. Substituting $(m,b)$ yields the policy frontier
\begin{equation}
\label{eq:frontier}
\quad
\kappa^\star(\mu)=1-\frac{\,1-\alpha(1-\delta)(1-\mu)\,}{(1-\delta)\mu}
\end{equation}
which is bounded in $[0,1]$, corresponding range of values to feasible regulatory intensities.
\end{prop}

The frontier $\kappa^\star(\mu)$ is derived from the reduced two-dimensional system. By Lemma~\ref{lemma1}, the asymptotic equivalence between the full $(n+1)$-dimensional system and the reduced system additionally requires $m\rho<1$ for some $\rho\in(|\lambda_2(S)|,1)$. When $\rho(A)<1$, each round of contagion and feedback is weaker than the preceding one. Thus, distortions are absorbed and eventually vanish. When $\rho(A)>1$, the combined effect of network contagion and AI retraining magnifies distortions, generating a self-reinforcing informational cascade. Small errors propagate, feed back into AI training, and return amplified in the next period. Here the term \emph{systemic contagion} refers to the fact that the distortions originating from AI not only diffuse through the social network but also feed back into the AI’s next training cycle, amplifying errors. This is the fundamental systemic informational risk created by the recursive use of AI. The regulator's task is then to push the system back into the stable zone by filtering a sufficient fraction $\kappa$ of errors originating from AI.

The convergence or divergence of $x_t$ occurs at exponential rate $\rho(A)^t$. If $\rho(A)$ is close to one, distortions decay very slowly; if $\rho(A)>1$, they grow explosively like $\rho(A)^t$. Hence, $\rho(A)$ plays a role analogous to the basic reproduction number in epidemiology, it is the fundamental measure of the strength of the contagion feedback loop.

\medskip

The regulator controls $\kappa$, the share of distortion originating from AI filtered out before it enters the social network. The policy frontier $\kappa^\star(\mu)$ defines the systemic-risk threshold: it is the minimal level of intervention required to prevent systemic informational collapse. When $\kappa<\kappa^\star(\mu)$, the feedback loop is too strong and the distortions amplify. Otherwise, whenever $\kappa\ge\kappa^\star(\mu)$, distortions remain bounded and eventually vanish. Remark that the analytical expression for $\kappa^\star(\mu)$ may fall outside the feasible range $[0,1]$ for extreme parameter values. Since $\kappa$ represents a share of filtering between 0 (no regulation) and 1 (perfect filtering), we project it onto the admissible domain $[0,1]$. Formally,\footnote{In the model $\mu\in[0,1]$ by definition. The projection is written to enforce the feasibility constraint $\kappa\in[0,1]$. The upper truncation may bind when the purely social propagation term is already explosive, i.e. when $m=(1-\delta)(1-\mu)\alpha>1$. In the restricted region $m\le 1$, the expression inside the projection is always weakly below one, so the projection could equivalently be written as a lower truncation at zero. We keep the full projection because we do not impose $m\le 1$ globally.}

$$
\kappa^\star(\mu)=\Pi_{[0,1]}\!\left(1-\frac{\,1-\alpha(1-\delta)(1-\mu)\,}{(1-\delta)\mu}\right)
:=\min\Big\{1,\max\Big\{0,\ 1-\frac{\,1-\alpha(1-\delta)(1-\mu)\,}{(1-\delta)\mu}\Big\}\Big\}.
$$

Figure \ref{fig1} illustrates the stability boundary  $\kappa^\star(\mu)$ separating safe informational regimes from systemic contagion.  The horizontal axis represents AI penetration $\mu$, and the vertical axis represents  the filtering intensity $\kappa$. On the interior region where the purely social system (or no-AI system, i.e., $\mu=0$) is stable, the frontier is upward-sloping: as content originating from AI becomes a larger share of the information environment, society becomes more dependent on AI, and distortions injected by AI have a larger system-wide impact. To offset this greater exposure, the regulator must filter a larger fraction of errors originating from AI. The region below the frontier corresponds to $\rho(A)>1$, where distortions propagate and reinforce each other; the region above corresponds to $\rho(A)<1$, where the combined AI–society system is stable and distortions vanish.

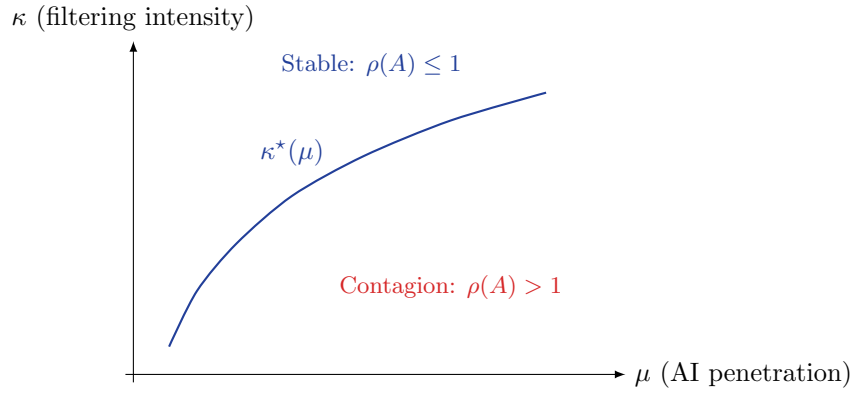
\begin{figure}[htbp!]
\centering
\begin{tikzpicture}[scale=1.05, >=latex]
\draw[->] (-0.1,0) -- (6.2,0) node[right] {$\mu$ (AI penetration)};
\draw[->] (0,-0.1) -- (0,4.2) node[above] {$\kappa$ (filtering intensity)};
\draw[thick, uniblau] plot[smooth] coordinates {
  (0.45,0.35)
  (0.80,1.05)
  (1.30,1.65)
  (2.00,2.25)
  (2.90,2.75)
  (4.00,3.20)
  (5.20,3.55)
};
\node[uniblau] at (3,3.9) {\small Stable: $\rho(A)\leq 1$};
\node[rot] at (4,1.1) {\small Contagion: $\rho(A)>1$};
\node[uniblau] at (2,2.8) {$\kappa^\star(\mu)$};
\end{tikzpicture}
\caption{The policy frontier $\kappa^\star(\mu)$: minimal filtering required to prevent systemic informational contagion as a function of AI penetration.}\label{fig1}

\end{figure}

The following Corollary makes explicit how the policy frontier depends on the AI penetration.

\begin{corollary}[Policy thresholds]\label{cor1}
For any $(\alpha,\mu,\delta,\gamma)$ there exists a threshold $\kappa^\star(\mu)$ such that
$\rho(A)\le 1$ if and only if $\kappa\ge \kappa^\star(\mu)$. Moreover:
\begin{itemize}
\item No systemic contagion: if $\kappa \ge \kappa^\star(\mu)$ then $\rho(A)\le 1$.
\item Systemic contagion: if $\kappa < \kappa^\star(\mu)$ then $\rho(A)>1$.
\item $\rho(A)$ is decreasing in $\kappa$ and $\delta$, and increasing in $\alpha$.
\item  Although $\rho(A)$ is not globally monotonic in $\mu$ or $\gamma$, the frontier $\kappa^\star(\mu)$ is increasing in $\mu$ on the interior region where the no-AI system (i.e.\ $\mu=0$) is stable, that is, whenever $\alpha(1-\delta) < 1$. In this region, higher AI penetration strengthens the AI-injection channel relative to the network attenuation effect and therefore requires stronger upstream filtering to preserve stability.
\end{itemize}
\end{corollary}

The boundary $\kappa^\star(\mu)$ in the plane $(\mu,\kappa)$ is nonlinear because both the contagion term $m=(1-\delta)(1-\mu)\alpha$ and the injection term $b=(1-\delta)\mu(1-\kappa)$ depend on $\mu$ in opposite directions. Increasing $\mu$ weakens the network–attenuation effect but strengthens AI error injection. 

\medskip

Because effective AI injection $b$ rises with $\mu$, higher AI penetration means a larger share of content originating from AI entering society. This increases reliance on AI-generated information and strengthens the feedback loop. Consequently, $\kappa^\star(\mu)$ increases with $\mu$: the more society depends on AI, the stronger the regulatory safeguard must be to preserve informational stability.

\medskip

The triplet $(\alpha,\mu,\gamma)$ jointly determines systemic-risk regimes. Increasing $\kappa$ or $\delta$ improves stability, while increasing $\alpha$ raises systemic vulnerability. The effects of $\mu$ and $\gamma$ on $\rho(A)$ are not globally monotonic: $\mu$ simultaneously increases AI injection and attenuates social propagation, while $\gamma$ simultaneously weakens feedback $(1-\gamma)$ and increases persistence of AI distortion. Nevertheless, on the economically relevant interior region, the stabilizing requirement $\kappa^\star(\mu)$ is increasing in $\mu$, because higher AI penetration strengthens the AI-injection channel relative to the network attenuation effect.

\begin{itemize}
\item AI Social-Information Network $\alpha$: Higher amplification $\alpha$ increases the social propagation term $m=(1-\delta)(1-\mu)\alpha$ and raises systemic vulnerability, requiring stronger filtering.
\item AI Contagion $\mu$: The effect of $\mu$ on $\rho(A)$ is not globally monotonic, since $\mu$ increases the injection term $b$ but decreases $m$. However, the stabilizing frontier $\kappa^\star(\mu)$ is increasing in $\mu$ on the interior region: higher AI penetration requires stronger filtering to offset greater injection originating from AI.
\item AI Multiplier $\gamma$: The multiplier effect is driven by the feedback weight $(1-\gamma)$. Lower $\gamma$ means stronger retraining on social distortion. At the same time, higher $\gamma$ makes AI distortion more persistent. Hence the effect of $\gamma$ on $\rho(A)$ is not globally monotonic.
\end{itemize}

In short, the regulation lessons from Propositions \ref{prop1} and \ref{prop2} and Corollary \ref{cor1} can be summarized as follow. When AI injects errors into society, the network spreads them, and the resulting distortions become part of the AI’s next training set. This loop repeats, and its stability hinges on a single eigenvalue, $\rho(A)$. The regulator's objective is to ensure $\rho(A)<1$: the frontier $\kappa^\star(\mu)$ identifies the minimum filtering required to prevent informational collapse, i.e., to avoid paths along which $\|x_t\|$ diverges due to the interaction between contagion and feedback.

\section{Network Topology and Systemic Informational Risk}
\label{homophily}

This section studies how network topology shapes systemic informational risk in the coupled AI-society system. While the aggregate dynamics are characterized by the $2\times2$ matrix $A(\mu,\kappa)$, the underlying network enters through distinct spectral margins. We distinguish two conceptually distinct effects. First, homophily affects the persistence of distortions by slowing the decay of subdominant network modes, without changing the long-run amplification parameter. Second, core-periphery structures affect the direction of contagion by concentrating the dominant eigenmode on a small set of highly influential nodes. We focus on asymptotic, structural results that characterize how these topologies shape long run stability and the policy frontier. Finite horizon implications and targeted regulation are analyzed separately in the Appendix Section \ref{finitehorizon}.

\medskip

A substantial empirical and theoretical literature documents that homophily (the tendency of individuals to interact more intensively with similar peers; \citealp{McPhersonSmith-LovinCook}) slows information diffusion and generates persistent segregation of beliefs. In the \cite{DeGroot} framework, the speed at which distortions converge toward the dominant eigenmode is determined by the spectral gap $1-|\lambda_2(S)|$. A smaller spectral gap (larger $|\lambda_2(S)|$) implies that subdominant modes decay more slowly \emph{relative to the dominant eigenmode}, causing distortions to remain locally persistent before eventually diffusing across the population.\footnote{For primitive row-stochastic matrices, convergence speed toward the stationary distribution is determined by the second-largest eigenvalue in modulus, see, e.g., \cite{DeGroot} and \cite{GolubJackson} (or results on mixing times of Markov chains).} In large random networks, \citet{GolubJackson} show that the degree-weighted homophily index $\mathrm{DWH}(S)$ approximates the second eigenvalue:
$$
|\lambda_2(S)| \approx \mathrm{DWH}(S),
$$
so that structural segregation directly maps into slow information aggregation. Economically, this means that homophilous societies or platforms (for instance due to algorithmic filtering, political segmentation, or linguistic communities) tend to retain pockets of misinformation for long periods. These distortions may later spill over across community boundaries, producing delayed but amplified contagion.

This role of homophily is consistent with recent results showing that network segregation creates learning bottlenecks and persistent disagreement even under rational updating \citep{AkbarpourMalladiSaberi, Sadler}. However, those contributions study environments in which information is aggregated only through social learning. 
In contrast, our framework introduces endogenous AI retraining, which feeds persistent local distortions into a centralized information-producing system. As a result, homophily does not merely slow convergence: it becomes a source of systemic informational risk.

\paragraph{Homophily and systemic risk boundary.} 
The effective amplification parameter in the AI-society system is $m=(1-\delta)(1-\mu)\alpha$, which measures how strongly distortions are amplified by the social network before encountering verification or replacement by AI content. Importantly, homophily does not increase the amplification parameter $\alpha$. Instead, homophily increases the persistence of distortions by raising $|\lambda_2(S)|$, while $\alpha$ scales the system-wide impact of these persistent distortions.

For a given level of AI penetration $\mu$, verification $\delta$, and amplification $\alpha$, a more homophilous network has two effects:

\begin{enumerate}
\item It slows the dissipation of distortions relative to the dominant eigenmode by reducing the spectral gap $1-|\lambda_2(S)|$.
\item It increases the likelihood that locally persistent distortions are repeatedly re-amplified
through AI retraining before they dissipate.
\end{enumerate}

Importantly, homophily does not alter the asymptotic stability condition $\rho(A(\mu,\kappa)) < 1$, which depends only on the amplification parameters $(\alpha,\mu,\delta,\gamma,\kappa)$ and not on $|\lambda_2(S)|$
under our normalization. Instead, higher homophily increases the persistence of subdominant network modes, so that distortions dissipate more slowly relative to the dominant eigenmode. As a consequence, cumulative distortions over economically relevant finite horizons are larger in more homophilous networks. This implies more stringent safety requirements at finite horizons, even though the asymptotic stability frontier $\kappa^{\star}(\mu)$ that determines long-run dynamics remains unchanged.

\newpage

\begin{prop}[Homophily and persistence of distortions]\label{prop3}
Fix $(\mu,\delta,\gamma,\alpha,\kappa)$ and let $S$ be primitive and row-stochastic. Let $W=\alpha S$ and consider the coupled AI-society dynamics. Then:
\begin{enumerate}
\item[(i)] The Perron eigenvalue of $S$ equals $1$ (hence the Perron eigenvalue of $W$ equals $\alpha$).
\item[(ii)] For any $\rho\in(|\lambda_2(S)|,1)$, the contribution of non-dominant network modes to aggregate distortions is bounded in absolute value by $C(m\rho)^t$, where $m=(1-\delta)(1-\mu)\alpha$. Equivalently, relative to the dominant eigenmode, the persistence of subdominant modes is characterized by $|\lambda_2(S)|$.
\item[(iii)] Consequently, higher homophily (larger $|\lambda_2(S)|$ or higher $\mathrm{DWH}(S)$ in large networks) implies slower dissipation of local distortions and slower convergence of the social information environment toward its dominant mode.
\end{enumerate}
\end{prop}

Proposition~\ref{prop3} shows that homophily does not merely slow down social learning, it increases the persistence of local distortions in the social–informational network. In highly homophilous societies, deviations from the dominant information mode decay slowly, so distorted signals remain trapped within clusters for extended periods. When AI systems retrain on the aggregate informational environment, this persistence is no longer innocuous: local distortions are repeatedly fed back into AI training and reinjected into society.

However, policy objectives are typically formulated over finite and economically relevant horizons. When AI systems retrain on the aggregate informational environment, slow dissipation of local distortions can translate into higher cumulative exposure over finite horizons, even when the asymptotic system is stable. Appendix~\ref{finitehorizon} formalizes this distinction and shows that, for any fixed horizon $T$, the finite-horizon stabilizing threshold $\kappa_T^\star(\mu;S)$ depends monotonically on $|\lambda_2(S)|$. So more homophilous networks require weakly stronger upstream filtering to maintain safety over policy-relevant horizons. Together, the asymptotic and finite-horizon results clarify how homophily affects both long-run stability and short to medium-run regulatory needs.

\medskip

Building on \citeauthor{GolubJackson}'s (\citeyear{GolubJackson}) analysis of how homophily slows information aggregation, recent work shows that network segregation creates learning bottlenecks and persistent disagreement even under rational updating \citep{AkbarpourMalladiSaberi,Sadler}. Those papers study how network segregation slows information aggregation in purely social learning environments. Our contribution is complementary but fundamentally distinct. We show that once AI systems retrain on the aggregate informational environment, persistence of local distortions is no longer innocuous: it becomes a source of \emph{systemic informational risk}. Homophily does not merely delay convergence, it prolongs the presence of distorted signals that are repeatedly fed back into AI training and reinjected into society. This feedback loop is absent from existing models of social learning and is central to the policy implications of AI-mediated information systems.

\medskip

Throughout the homophily analysis, aggregate distortion continues to be measured by the uniform average $\bar e_t = \frac{1}{n}\mathbf{1}^\top e_t$. Homophily affects the persistence with which distortions converge toward the dominant eigenmode, but does not alter the support of that mode itself. In contrast, core-periphery structures modify the distribution of asymptotic influence weights. In the rest of this section, we therefore measure aggregate distortion using the stationary distribution $\pi$ in order to capture directional informational fragility in centralized networks.

\paragraph{Core-Periphery Networks.} A second network topology particularly relevant for AI-mediated information environments is the core-periphery structure \citep{BorgattiEverett}. Many real-world information systems exhibit this architecture: a small, densely interconnected core of agents (major platforms, ranking algorithms, influential accounts, news aggregators) interacts intensively within itself and broadcasts information to a large periphery with weak internal connectivity.

\medskip

Let $S$ denote the row-stochastic and primitive attention matrix introduced in the model Section~\ref{model}. Partition the population into a core $C$ and a periphery $P$. To describe the core-periphery architecture, we introduce matrices $S_{CC},S_{CP},S_{PC},S_{PP}$ and represent the attention structure in the following stylized form:

$$
S=
\begin{bmatrix}
S_{CC} & \epsilon S_{CP}\\
S_{PC} & S_{PP}
\end{bmatrix},
$$
with $\epsilon\in[0,\bar\epsilon)$, where $\bar\epsilon>0$. To avoid confusion, $S_{CC}$ should not be interpreted as the literal upper-left block of the row-stochastic matrix $S$. Rather, $S_{CC}$ denotes the \emph{renormalized within-core attention matrix}, obtained by restricting attention to core-to-core interactions and renormalizing rows to sum to one. Assumption~\ref{ass:CP}(i) therefore states that the core, viewed as an autonomous subsystem, is internally connected and informationally cohesive.

Given our convention that $s_{ij}$ is the share of attention that agent $i$ allocates to agent $j$, rows index listeners and columns index sources. Hence $S_{PC}$ captures attention that peripheral agents allocate to the core (periphery listening to the core), while $S_{CP}$ captures attention allocated by core agents to the periphery. Under this parametrization, $S_{CC}$ captures dense interactions within the core, $S_{PC}$ captures outward broadcasting from the core to the periphery, and $\epsilon S_{CP}$ represents weak feedback from the periphery to the core. The effective propagation operator is $W=\alpha S$. The matrix $S_{CC}$ is an auxiliary renormalized object used only to describe within-core connectivity. All spectral properties and aggregation results refer to the original row-stochastic matrix $S$. We formalize a stylized but empirically relevant core-periphery structure with the following assumption:

\begin{ass}[Network structure]\label{ass:CP}
We assume:
\begin{enumerate}
\item[(i)] The within-core attention matrix $S_{CC}$ and the within-periphery matrix $S_{PP}$ are primitive, so that each subsystem is internally connected.
\item[(ii)] As $\epsilon \to 0$, the \emph{stationary distribution} (left Perron eigenvector)
$\pi(\epsilon)$ of $S$ concentrates on the core $C$, i.e., $\sum_{i\in C}\pi_i(\epsilon)\to 1$.\footnote{Throughout, we interpret the stationary
distribution $\pi$ (the left Perron eigenvector of $S$) as agents' long-run influence weights in social communication. Core-periphery structure affects outcomes through the concentration of $\pi$, not through the Perron root, which remains equal to one under row-stochastic normalization.}
\item[(iii)] Core-to-periphery attention dominates periphery-to-core attention, so that
$\epsilon\ll 1$.
\end{enumerate}
\end{ass}

Condition (i) ensures that both the core and the periphery are internally connected,
while allowing the core to be more informationally cohesive. Condition (ii) is the key identifying restriction: it states that, as feedback from the periphery to the core weakens, the stationary distribution of the attention matrix places increasing weight on core agents. Economically, this means that long-run aggregate beliefs and distortions are shaped disproportionately by a small set of highly influential core agents, platforms, or algorithms.

Condition (iii) guarantees directional influence: information flows primarily from the core to the periphery rather than symmetrically. When $\epsilon\ll1$, the information flow is predominantly one-directional: the informational core broadcasts information outward, while the periphery has limited influence on the core. This asymmetry ensures that the stationary influence weights of $S$ are concentrated on the core, so that long-run aggregate distortions are driven primarily by dynamics within the informational core.\footnote{If $\epsilon$ were large, periphery to core feedback could substantially alter the stationary influence weights (and therefore the dominant eigenmode), blurring the distinction between core and periphery. In that case, systemic informational risk would no longer be driven by a small set of high-centrality nodes, and the clean characterization of the stability frontier in terms of $\alpha$ would fail. The small $\epsilon$ assumption therefore isolates the empirically relevant regime of directional influence, which captures centralized platforms and recommender systems that broadcast widely but receive relatively little feedback from end users.} Together, these conditions ensure that the dominant eigenmode captures centralized informational control, allowing us to link network topology directly to the distribution and direction of systemic informational risk.

Unlike homophily, which affects the \emph{relative persistence} of distortions through the second eigenvalue $|\lambda_2(S)|$, a core-periphery structure affects the \emph{support} of the dominant eigenmode. When attention is centralized, the stationary distribution $\pi$ of $S$ becomes concentrated on a subset of high-centrality (core) nodes, so that distortions originating there dominate long-run aggregate dynamics. 

Here, we therefore measure aggregate distortion using
$$\bar e_t := \pi^\top e_t,$$
rather than the uniform average used in Sections \ref{model}-\ref{MainResults}. This aggregation reflects the asymptotic influence weights characterized by $\pi$ and is the relevant object for analyzing directional informational fragility in centralized networks. This change in aggregation does not alter the underlying dynamics or stability conditions, since the reduced two-dimensional system remains closed and continues to be determined by the same spectral properties of $S$ and $W$. Assumption~\ref{ass:CP} affects the distribution of long-run influence weights through the stationary distribution $\pi$, but does not alter the Perron eigenvalue of $S$ or the amplification parameter $\alpha$.

\begin{prop}[Core-periphery and directional informational fragility]
\label{prop4}
Suppose the attention matrix $S$ is primitive and row-stochastic and satisfies Assumption~\ref{ass:CP}, and let $W=\alpha S$.
Then $\rho(W)=\alpha$ for all $\epsilon$, and the asymptotic stability frontier is
$$
\kappa^\star(\mu)=1-\frac{1-\alpha(1-\delta)(1-\mu)}{(1-\delta)\mu}.
$$
Moreover:
\begin{enumerate}
\item[(i)] As $\epsilon\to 0$, the stationary distribution $\pi(\epsilon)$ concentrates on the core, i.e.
$\sum_{i\in C}\pi_i(\epsilon)\to 1$. Hence distortions originating in the core receive disproportionately large long-run influence weight in aggregate outcomes.
\item[(ii)] Under a fixed filtering budget, targeted regulation that filters high-$\pi_i$ agents (asymptotically, core agents) yields strictly larger reductions in systemic informational risk than uniform filtering.
\end{enumerate}
\end{prop}

Because the original attention matrix $S$ remains row-stochastic for all admissible values of $\epsilon$, its Perron eigenvalue is identically equal to one. Hence the spectral amplification factor $\rho(W)=\alpha$, and the asymptotic stability frontier $\kappa^\star(\mu)$ is independent of the core-periphery parameter $\epsilon$.

Core-periphery networks generate \emph{directional informational fragility}. Even if most agents are peripheral and individually cautious, distortions produced or amplified within the informational core dominate aggregate learning and are repeatedly fed back into AI retraining. Errors originating in the periphery, by contrast, tend to dissipate locally and have little long-run impact. From a regulatory perspective, this implies that effective AI governance cannot be uniform across agents or platforms. When informational influence is centralized, systemic risk is driven by the structure and behavior of high-centrality platforms and nodes rather than by average users. Targeted regulation of high-centrality platforms, ranking algorithms, and influential nodes is therefore essential to prevent systemic informational contagion.

From a policy perspective, core-periphery architectures imply that effective AI regulation cannot be uniform across agents or platforms. When informational influence is centralized, 
systemic risk is driven by the amplification properties of the informational core rather than by average user behavior. Regulatory efforts should therefore focus on high-centrality nodes (such as major platforms, ranking algorithms, or influential accounts) whose distortions disproportionately shape the aggregate signal on which AI systems retrain. Monitoring peripheral users alone is insufficient to prevent systemic informational contagion.

Appendix \ref{finitehorizon} shows that these asymptotic insights extend to finite-horizon policy objectives: when filtering can be targeted, concentrating regulatory effort on the most influential nodes strictly dominates uniform regulation under realistic budget constraints.

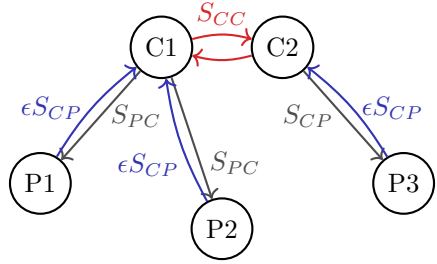
\begin{figure}[h!]
\centering
\begin{tikzpicture}[
    core/.style={circle, draw, thick, minimum size=18pt},
    peri/.style={circle, draw, thick, minimum size=16pt},
    edgeC/.style={->, thick, rot},
    edgeP/.style={->, thick, blau},
    edgeCP/.style={->, thick, black!70},
    scale=1
]
\node[core] (c1) at (-0.8,0.6) {C1};
\node[core] (c2) at (0.8,0.6) {C2};
\node[peri] (p1) at (-2.4,-1.2) {P1};
\node[peri] (p2) at (0,-1.8) {P2};
\node[peri] (p3) at (2.4,-1.2) {P3};
\draw[edgeC] (c1) to[bend left=15] node[midway,above] {$S_{CC}$} (c2);
\draw[edgeC] (c2) to[bend left=15] (c1);
\draw[edgeCP] (c1) -- node[midway,right] {$S_{PC}$} (p1);
\draw[edgeCP] (c1) -- node[pos=0.7,right] {$S_{PC}$} (p2);
\draw[edgeCP] (c2) -- node[midway,left] {$S_{CP}$} (p3);
\draw[edgeP] (p1) to[bend left=10] node[midway,left] {$\epsilon S_{CP}$} (c1);
\draw[edgeP] (p3) to[bend right=10] node[midway,right] {$\epsilon S_{CP}$} (c2);
\draw[edgeP] (p2) to[bend left=10] node[pos=0.3,left] {$\epsilon S_{CP}$} (c1);
\end{tikzpicture}
\caption{Core-periphery attention structure (illustration of $S$)}

\medskip

{\footnotesize Notes: Nodes labeled $C_1, C_2$ represent agents in the core, while nodes labeled $P_1, P_2, P_3$ represent agents in the periphery. Directed edges correspond to attention weights from the row-stochastic matrix $S$: arrows point from information sources to listeners. The core exhibits dense internal connectivity ($S_{CC}$), while links from the core to the periphery ($S_{PC}$) are strong and feedback from the periphery to the core ($\epsilon S_{CP}$) is weak.}
\label{fig:core_periphery}
\end{figure}

Figure~\ref{fig:core_periphery} provides a schematic interpretation of the core-periphery network structure formalized in Assumption~\ref{ass:CP} and analyzed in Proposition~\ref{prop4}. Directed links represent attention weights from the row-stochastic matrix $S$. The core exhibits dense mutual attention ($S_{CC}$) and strong outward influence ($S_{PC}$), while feedback from the periphery to the core ($\epsilon S_{CP}$) is weak. As $\varepsilon\to 0$, the stationary distribution (left Perron eigenvector) of $S$ concentrates on the core. Consequently, long-run aggregate distortions are driven primarily by informational dynamics within the informational core, even if the periphery contains the majority of agents. Distortions originating in peripheral nodes tend to dissipate locally, whereas distortions generated or amplified within the informational core dominate aggregate learning and are repeatedly fed back into AI retraining. The figure highlights the resulting directional informational fragility. Systemic informational risk is driven by the structure and behavior of high-centrality nodes rather than by average behavior across the network. The scalar $\alpha$ captures global amplification while the concentration of influence is determined by network topology.

\paragraph{Policy Implications: Network-Aware AI Regulation}

The homophily and core-periphery results have direct implications for the design of AI regulation. Systemic informational risk is not determined solely by the volume of AI-generated content or the intrinsic accuracy of AI systems, but also by the topology of the social-information network into which AI outputs are injected.

\medskip

In homophilous or polarized societies, slow information aggregation increases the persistence of local distortions relative to the dominant mode. Even when AI systems are only mildly biased, these distortions remain clustered long enough to be repeatedly reamplified through AI retraining. As a result, over policy-relevant finite horizons the stabilizing filter requirement $\kappa_T^\star(\mu;S)$ is higher in homophilous environments, even though the asymptotic frontier $\kappa^\star(\mu)$ is topology-invariant under our normalization. Regulatory standards for AI-generated content should therefore be stricter in segmented or polarized informational environment than in more mixing societies.

In centralized information architectures, systemic informational risk is driven disproportionately by the informational core. Errors originating in highly connected platforms, ranking algorithms, or influential accounts dominate the stationary distribution (left Perron eigenvector) and thus the long-run dynamics of misinformation. Regulation that targets only average behavior is insufficient. Effective policy must focus on high-centrality nodes and central information-processing systems, providing a formal justification for platform-specific obligations, asymmetric compliance requirements, and enhanced scrutiny of algorithmic ranking systems. Importantly, AI retraining transforms slow social learning into a feedback mechanism that amplifies network fragilities. While homophily creates learning bottlenecks in purely social settings, AI retrains on the aggregate signal generated by these bottlenecks, turning persistence into systemic instability. Regulation must therefore act \emph{upstream}, before distortions propagate through the network and re-enter AI training.

\medskip

To sum up, optimal AI regulation must be network-aware. While the asymptotic stability frontier $\kappa^\star(\mu)$ depends on average amplification and exposure parameters, network topology crucially shapes \emph{systemic informational risk in practice} by determining (i) the persistence of distortions over policy-relevant horizons (homophily, via $|\lambda_2(S)|$) and (ii) which nodes drive long-run influence and therefore where regulation should be targeted (core-periphery, via the concentration of $\pi$). Regulatory standards that ignore topology may be ineffective in fragile architectures or excessively costly in well-mixing ones.

\section{A Simple Microfoundation of Verification and Stability}\label{sec:micro}

This section provides a microfoundation for the verification efforts that mitigate misinformation originating from AI. We discuss a microfounded network game of verification efforts with local complementarities, then derive its Nash equilibrium and link it to the macro diffusion system that determines the dynamics of misinformation. Our ordering (equilibrium first, then macro stability) keeps identification and policy interpretation transparent. Agents exercise costly verification effort $a_i$ to mitigate exposure to AI-generated distortions, but the benefits of verification depend on neighbors' own verification choices. This creates strategic complements and local spillovers (verification is more productive when neighbors verify too), exactly the setting in which Katz–Bonacich centrality, equilibrium existence, and monotone comparative statics apply.

\paragraph{A linear-quadratic network game.}

There are $n$ informational agents $i=1,\ldots,n$ connected by a weighted directed attention network $S=(s_{ij})$, with $s_{ij}\geq 0$ and $\sum_j s_{ij}=1$. As in the model section (Section \ref{model}), we separate the pattern of attention $S$ from the overall interaction intensity. We let $\alpha>0$ denote an amplification/intensity parameter and define the effective interaction matrix in the game as $W:=\alpha S$, so that $\rho(W)=\alpha$. Each agent chooses a nonnegative verification effort $a_i\ge 0$ (fact-checking, abstention from low-quality content, local moderation). Let $a=(a_1,\ldots,a_n)^\top$ and $\bar a = \frac{1}{n}\sum_i a_i$. The informational environment is summarized by two primitives in this static game: the AI penetration rate $\mu\in[0,1]$ and the AI distortion level $q\ge 0$. We treat $q$ as a parameter within the one-shot game; its dynamic law $q_{t+1}=\gamma q_t+(1-\gamma)\bar e_t$ is determined by the macro system. Altogether, the linear-quadratic network game with strategic complements and local spillovers defines a specific \emph{verification game}.

\medskip

Verification yields benefits both directly and through neighbors’ efforts, but is costly and is motivated by exposure to errors originating from AI. For parameters $\theta>0$ (baseline value of accuracy), $\phi>0$ (strategic complementarity), $c>0$ (effort cost), and $\lambda>0$ (exposure loss), the payoff of agent $i$ is
\begin{equation}
u_i(a;W,\mu,q)=\underbrace{\theta\,a_i + \phi a_i \sum_j w_{ij} a_j}_{\text{verification benefits with local complements}}-\underbrace{\frac{c}{2}a_i^2}_{\text{convex cost}} - \;\; \lambda\,\mu q\, a_i
\label{eq:payoff_micro}
\end{equation}

The last term $-\lambda \mu q\,a_i$ is a reduced-form way to allow the verification environment to depend on AI penetration and AI distortion. It should be interpreted as an environment-dependent component of the (opportunity) cost of verification, or equivalently as a reduction in the marginal returns to effort when the AI environment is more distorted. For example, when AI-generated content is noisier (high $q$) or more prevalent (high $\mu$), verification becomes more time-consuming, harder to implement, or less reliably rewarded. This formulation yields the empirically plausible comparative statics that equilibrium verification falls with $\mu q$, which then feeds into the macro system through a lower effective decay rate $\delta_{LQ}$ (defined below as the verification-adjusted decay rate of distortions).

This is the standard linear-quadratic network game form used in \citet{BallesterCalvoArmengolZenou, GaleottiGoyalJacksonVega-RedondoYariv, JacksonZenou2014}. It is also analogous to the communication-precision choices in \citet{GaleottiGhiglinoSquintani}, where agents choose how precisely or truthfully to transmit information in a network.

\medskip

The first-order condition for $u_i$ yields a linear best reply:
$$
\frac{\partial u_i}{\partial a_i}= \theta + \phi\sum_j w_{ij}a_j - c a_i - \lambda\mu q=0
\quad\Longrightarrow\quad
a_i 
= 
\frac{\theta - \lambda \mu q}{c}
+
\frac{\phi}{c} \sum_j w_{ij} a_j.
$$
Define
$$
\beta = \frac{\theta - \lambda\mu q}{c} \text{ and } \psi=\frac{\phi}{c}.
$$
In vector form,
\begin{equation}
a = \beta \mathbf{1} + \psi W a. 
\label{eq:BR_micro}
\end{equation}
If $\psi \rho(W) < 1$, the linear system \eqref{eq:BR_micro} has a unique solution
\begin{equation}
a^\star = (I - \psi W)^{-1} \beta \mathbf{1} = \beta\sum_{k=0}^\infty (\psi W)^k \mathbf{1}.
\label{eq:Bonacich_micro}
\end{equation}
Thus equilibrium verification is proportional to Katz–Bonacich centrality in $W$:
more central agents (who are more influential or more exposed) verify more on average
\citep{BallesterCalvoArmengolZenou,JacksonZenou2014}. Proposition \ref{prop:NE_micro} sum up these findings:

\begin{prop}[Existence, uniqueness, and monotonicity]\label{prop:NE_micro}
If $\psi\rho(W)<1$, the verification game admits a unique Nash equilibrium $a^\star=(I-\psi W)^{-1}\beta\mathbf{1}$. Moreover:
\begin{enumerate}
\item $a^\star_i$ is increasing in the Katz–Bonacich centrality of $i$ in $W$;
\item $a^\star$ is increasing in $\psi$ and entrywise in $W$, and decreasing in $\mu q$;
\item the average equilibrium effort $\bar a^\star$ is increasing in $\theta,\phi$ and decreasing in 
$c,\lambda,\mu,q$.
\end{enumerate}
\end{prop}

These properties mirror general results for network games with strategic complements \citep{GaleottiGoyalJacksonVega-RedondoYariv,JacksonZenou2014}: denser or more central neighbors raise verification, while greater AI penetration or worse AI quality depress it.

\paragraph{Micro verification to macro contagion.}

Verification increases the effective removal rate of misinformation in the macro system. We capture this with the usual two-scale closure:\footnote{We use the subscript $LQ$ to represent the parameter $\delta$ by $\delta_{LQ}$ in the linear-quadratic network game.}
\begin{equation}
\delta_{LQ} = \delta_0 + \eta \bar a^\star, \;\; \bar a^\star = \frac{1}{n}\sum_i a^\star_i, \;\; \eta>0,
\label{eq:delta_eff}
\end{equation}
so stronger equilibrium verification raises the effective decline of distortions. Substituting \eqref{eq:delta_eff} into the macro coefficients of the 2×2 system yields
$$
m = (1-\delta_{LQ})(1-\mu)\alpha = \bigl(1-\delta_0-\eta\bar a^\star\bigr)(1-\mu)\alpha \text{ and } b = (1-\delta_{LQ})\mu(1-\kappa) = \bigl(1-\delta_0-\eta\bar a^\star\bigr)\mu(1-\kappa).
$$

\begin{prop}[Micro efforts and macro amplification]\label{prop:micro_macro}
Consider the unique Nash equilibrium $a^\star$ and $\delta_{LQ},m,b$ defined as above. Then:
\begin{enumerate}
\item If the network strengthens ($W'\ge W$ entrywise) or complements rise ($\phi$ increases), then $\bar a^\star$ increases, so $\delta_{LQ}$ increases and both $m$ and $b$ decrease (holding $\mu,\kappa,\alpha$ fixed).
\item If AI penetration $\mu$ or AI distortion $q$ rises, then $\bar a^\star$ falls and $\delta_{LQ}$ falls. As a result, the AI-injection coefficient $b$ increases. The effect on $m$ is ambiguous in general: $m$ increases whenever the verification response is sufficiently elastic so that the rise in $(1-\delta_{LQ})$ dominates the mechanical decline in $(1-\mu)$.

\end{enumerate}
\end{prop}

Proposition~\ref{prop:micro_macro} formalizes how micro-level verification incentives shape the macro-level contagion coefficients. Stronger verification, arising from denser networks or stronger strategic complementarities, raises the effective decay rate of misinformation and therefore reduces both the internal amplification term $m$ and the AI-injection term $b$. By contrast, higher AI penetration or worse AI quality weaken verification incentives, lowering the effective verification rate $\delta_{LQ}$. This unambiguously increases the AI-injection coefficient $b$, strengthening the feedback from AI into society. The effect on the internal amplification term $m$ is in general ambiguous: higher AI penetration mechanically attenuates human-to-human diffusion through $(1-\mu)$, but this effect may be outweighed by the reduction in verification when verification incentives are sufficiently sensitive. Taken together, the proposition shows that verification behavior is an important transmission mechanism through which AI penetration and network structure influence systemic informational risk.

\medskip

The regulator chooses a filtering intensity $\kappa\in[0,1]$ (strength of the KL-based I-projection) to reduce distortions of AI origin before they enter the network. The macro analysis in Section~\ref{2dim} shows that stability of the reduced two-dimensional system is equivalent to $\rho(A(\mu,\kappa))\le 1$ (and, in the stable regime, the full system is asymptotically characterized by the same condition), and the minimal stabilizing filter is given by the frontier

$$
\kappa^\star(\mu)
= 1-\frac{1-\alpha(1-\delta_{LQ})(1-\mu)}{(1-\delta_{LQ})\mu}
$$
projected onto $[0,1]$. This frontier characterizes the boundary separating stable from systemic informational regimes.

We now define a simple \emph{regulation equilibrium}: the regulator chooses a minimally costly $\kappa$ on or above this frontier while agents play the Nash equilibrium $a^\star(W,\mu,q)$. 
We collect the regulator concern about misinformation in a reduced-form loss function $L(\kappa)$, which represents the expected welfare loss from the path of distortions generated by $x_{t+1}= A(\mu,\kappa)x_t$. We do not pin down a specific functional form; for example, $L(\kappa)$ could be the discounted sum of squared distortions (a quadratic loss) when $\rho(A(\mu,\kappa))<1$, and very large (or infinite) when $\rho(A(\mu,\kappa))\ge 1$. We assume $L$ is continuous on $[0,1]$ and finite on the stability region $\{\kappa:\rho(A(\mu,\kappa))<1\}$. The cost of filtering, $C(\kappa)$, captures implementation and compliance costs of regulatory interventions, which also request that $C$ is continuous on $[0,1]$.\footnote{Additional structure such as convexity of $C$ will guarantee uniqueness.}

\begin{definition}
A regulation equilibrium is a pair $(a^\star,\kappa^{\star\star})$ such that:
\begin{enumerate}
\item $a^\star$ is the unique Nash equilibrium of the verification game characterized by the payoff function \eqref{eq:payoff_micro}.
\item $\kappa^{\star\star}$ corresponds to the optimal regulator policy:
$\min_{\kappa\in[0,1]} \;\big\{\, L(\kappa) + C(\kappa)\,\big\}
\quad\text{s.t.}\quad \kappa \ge \kappa^\star(\mu).$
\end{enumerate}
\end{definition}

Our regulation equilibrium addresses a network externality analogous to those identified by \cite{BlochDemange}: individually optimal actions fail to internalize system-wide informational spillovers, leading to inefficient outcomes.

\begin{prop}[Regulation equilibrium]
Suppose $L(\cdot)$ and $C(\cdot)$ are continuous on $[0,1]$. Then a regulation equilibrium exists. Whenever the stability constraint binds, the minimal-cost 
regulation is $\kappa^{\star\star}=\kappa^\star(\mu)$. It follows:
\begin{itemize}
\item (Ceteris paribus) Holding $\delta_{LQ}$ fixed, $\kappa^{\star\star}$ is increasing in $\alpha$ and in $\mu$.
\item $\kappa^{\star\star}$ is decreasing in $\delta_{LQ}$, and hence decreasing in $\bar a^\star$.
\end{itemize}
\end{prop}

When $\delta_{LQ}$ is endogenously determined by the verification game, changes in $\alpha$ affect $\kappa^{\star\star}$ both directly (through network amplification) and indirectly (through higher equilibrium verification and thus higher $\delta_{LQ}$). The total effect of $\alpha$ on $\kappa^{\star\star}$ is therefore, in general, ambiguous without further restrictions on the elasticity of $\delta_{LQ}$ with respect to $\alpha$. The comparative statics of $\kappa^\star(\mu)$, and therefore of $\kappa^{\star\star}$, with respect to $(\alpha,\mu,\delta_{LQ})$ follow directly from Proposition~\ref{prop:micro_macro} together with the spectral stability condition characterized in Section~\ref{MainResults}.

While stability refers to the dynamic macro condition $\rho(A(\mu,\kappa))<1$, a regulation equilibrium is a combination of $(i)$ behavior in verification at the Nash equilibrium, and $(ii)$ the cheapest feasible regulatory intensity on or above the spectral frontier. Network structure and AI penetration determine both: stronger amplification (higher $\alpha$) and heavier AI use (higher $\mu$) push the system toward instability and raise the minimal stabilizing filter $\kappa^\star(\mu)$. By contrast, homophily affects the persistence of distortions and therefore raises finite-horizon stabilization requirements, even though the asymptotic frontier $\kappa^\star(\mu)$ is unchanged under our normalization.\footnote{Appendix~\ref{finitehorizon} formalizes the monotone dependence of the finite-horizon stabilizing threshold on $|\lambda_2(S)|$.} At the same time, stronger verification incentives (higher $a^\star$) raise $\delta_{LQ}$ and relax the regulatory burden. This mirrors the interplay between micro risk-taking and macroprudential regulation in financial networks  \citep{ElliottGolubJackson}: systemic fragility arises from feedback interactions between individual incentives, network structure, and central interventions. Here the object of concern is informational rather than financial stability, but the mathematical logic and policy trade-offs are analogous.

\section{Concluding Discussion}\label{conclusion}

This paper analyzes how artificial intelligence operating within social communication networks affects the stability of collective knowledge. The main insight is that AI systems are not merely external sources of information: they simultaneously shape the informational environment and retrain on data generated by that environment. This coupling creates a feedback loop between social diffusion and algorithmic learning. We show that this loop can generate \emph{systemic informational risk}, a regime in which small informational distortions are endogenously amplified over time. Despite the high dimensionality of the underlying system, its long-run behavior can be characterized by a low-dimensional structure: distortions align with a dominant amplification mode, while subdominant network effects decay geometrically relative to it at a rate characterized by $|\lambda_2(S)|$. In the stable regime, the high-dimensional AI–society system is asymptotically equivalent to a two-dimensional dynamic system whose spectral radius determines whether distortions are absorbed or amplified.

A theoretical contribution is an asymptotic dimensional reduction of the coupled AI-society system. Despite an $(n+1)$-dimensional environment with heterogeneous agents and network interactions, the long-run dynamics admit an asymptotically equivalent two-dimensional representation in average social distortion and AI distortion. This reduction follows from the Perron-Frobenius structure of social networks and from the fact that AI retraining responds to aggregate rather than individual signals. As a result, the stability of collective knowledge is determined by a single spectral condition. When the feedback loop between social contagion and AI retraining is sufficiently strong, informational distortions are endogenously amplified over time, producing systemic informational collapse driven by explosive informational dynamics. The model delivers three results of direct relevance for AI governance.

\paragraph{First, systemic informational risk is an upstream phenomenon.} Instability arises at the point where AI-generated content enters social information flows and is subsequently absorbed into future training data. Downstream interventions (ex-post content moderation, user-level fact checking, or individual corrections) are insufficient once distortions have entered the retraining loop. In the model, stability is restored only by reducing the effective injection of distortion originating from AI into society. This provides a formal justification for regulatory approaches that emphasize ex-ante obligations, such as training-data governance, model evaluation, and output filtering. This logic aligns closely with the risk-based architecture of the EU AI Act, which imposes the strongest obligations on systems classified as high risk due to their potential for large-scale social impact.\footnote{\url{https://eur-lex.europa.eu/legal-content/EN/TXT/?uri=CELEX:52021PC0206} (accessed 2026-01-22).} Our results clarify why such upstream regulation is necessary: systemic informational risk does not depend on malicious intent or individual misuse, but on the structural interaction between AI retraining and networked social learning. The relevant policy object is not isolated harm, but the stability of the informational environment as a whole.

\paragraph{Second, effective AI regulation must be network-aware.}
While the asymptotic stability frontier depends only on average amplification and exposure parameters under our normalization, network topology crucially shapes systemic informational risk in practice. In particular, homophily increases informational fragility by slowing the decay of subdominant network modes, thereby raising cumulative exposure over policy-relevant horizons. Appendix~\ref{finitehorizon} shows that the finite-horizon stabilizing
threshold depends monotonically on $|\lambda_2(S)|$: more homophilous networks require stronger upstream filtering to maintain safety over economically relevant time horizons, even when asymptotic stability is unchanged.

We also show that core-periphery structures generate directional fragility. When informational influence is concentrated, distortions originating in high-centrality nodes disproportionately shape aggregate learning and AI retraining. Systemic informational risk is therefore driven not only by the magnitude of AI-generated distortions, but by where in the network they originate. Uniform regulation is inefficient in such environments. Instead, effective governance must focus on high-centrality platforms, information intermediaries, and ranking infrastructures, providing a formal justification for platform-specific obligations, asymmetric compliance requirements, and enhanced scrutiny of algorithmic ranking systems.

These results provide a theoretical foundation for the systemic-risk provisions of the Digital Services Act (DSA).\footnote{\url{https://eur-lex.europa.eu/eli/reg/2022/2065/oj} (accessed 2026-01-22).} The DSA requires very large online platforms and search engines to assess and mitigate ``systemic risks'' arising from the dissemination of content. Our framework makes precise what such risks entail in an AI-mediated environment: platforms that occupy central positions in information networks, through ranking, recommendation, or search, play a disproportionate role in the feedback loop that determines informational stability. Uniform user-level regulation is therefore inefficient. Instead, governance should focus on high-centrality systems and channels, consistent with the DSA's asymmetric obligations and auditing requirements.

\paragraph{Third, verification and regulation are complements rather than substitutes.} We provide a microfoundation in which agents choose costly verification effort in a network game with strategic complementarities. Equilibrium verification determines the effective decay rate of misinformation in the macro system. Because verification incentives are shaped by network spillovers, they are generally insufficient to ensure stability on their own. Regulatory filtering of AI outputs raises the returns to verification and reduces amplification, while higher verification relaxes regulatory requirements. This complementarity suggests that policy frameworks should not treat ``user responsibility'' or ``media literacy'' as substitutes for upstream oversight. From a U.S. policy perspective, this logic connects naturally to the NIST AI Risk Management Framework, which emphasizes continuous risk assessment, monitoring, and mitigation across the AI lifecycle.\footnote{\url{https://www.nist.gov/itl/ai-risk-management-framework} (accessed 2026-01-22).} Our analysis identifies a specific class of risks (systemic informational instability) that are not captured by standard performance metrics but arise from interaction effects between AI systems and social networks. Incorporating network exposure and retraining feedback into risk-management processes would substantially strengthen the framework’s ability to address societal-scale harms.

More broadly, the paper contributes to the economics of AI governance by providing a formal stability theory for AI-mediated informational environments. While much of the existing literature focuses on bias, fairness, explainability, or isolated misinformation events, our analysis highlights a distinct threat: endogenous instability generated by feedback between diffusion and generation of information. The main objects we derive --stability thresholds, amplification factors, and regulatory frontiers-- are informational analogues of systemic risk concepts long studied in financial economics.

Several extensions are natural. First, recent work in computer science has emphasized the dangers of recursive training through model collapse \citep{Shumailov}. Our analysis complements this literature by showing that recursive feedback between AI and society may generate systemic informational instability even when AI systems themselves remain stable. More generally, translating the reduced-form filtering parameter into specific technical and institutional instruments, such as auditing standards, dataset constraints, or alignment requirements, remains an important empirical and design task. Second, richer nonlinearities in attention, platform incentives, or retraining frequency may generate additional path dependence and multiple equilibria. Third, measuring network exposure and centrality in real-world AI-mediated information systems is a prerequisite for operationalizing network-aware regulation.

To sum up, the analysis shows that systemic informational risk is not an accident of misuse or manipulation, but a predictable outcome of scale, connectivity, and feedback in AI-mediated societies. Risk-based, upstream, and asymmetric regulatory frameworks, such as those embodied in the EU AI Act and the DSA, are therefore not merely precautionary tools. They are structural requirements for preserving the stability of collective knowledge in the presence of adaptive AI systems integrated into social networks.

\section{Proofs}\label{appendix}

Throughout the Appendix and the results, $\|\cdot\|$ denotes an arbitrary submultiplicative matrix norm, i.e.\ $\|AB\|\le\|A\|\|B\|$. For vectors, $\|\cdot\|$ denotes the associated vector norm (e.g.\ Euclidean norm). Since the state space is finite-dimensional, all vector norms on $\mathbb{R}^n$ are equivalent. Hence the exponential bounds derived below hold under any fixed choice of norm, up to multiplicative constants.

\subsection{Known Intermediate Linear-Algebra Results}

We recall here known intermediate results in linear-algebra, and explain why and how they apply. These results, which are provided without proof, are taken from the textbook by \cite{HornJohnson}.

\begin{lemma}[Gelfand formula]\label{gelfand}
For any matrix $B$, $\displaystyle \lim_{t\to\infty}\|B^t\|^{1/t}=\rho(B)$.
\end{lemma}

It follows that the asymptotic rate of $B^t$ is determined by $\rho(B)$.

\begin{lemma}[Neumann series]\label{neumann}
If $\rho(B)<1$, then $(I-B)$ is invertible and $\sum_{s=0}^\infty B^s=(I-B)^{-1}$.
\end{lemma}

Combining this result and Lemma ~\ref{gelfand},  $\rho(B)<1$ implies $B^t\to 0$ and then the geometric series converges in operator norm to $(I-B)^{-1}$. We provide a proof of this result in Lemma \ref{lem:decay} hereafter.

\begin{lemma}[Perron-Frobenius]\label{pf}
If $W\ge 0$ is primitive, then $\rho(W)$ is a simple eigenvalue with $v\gg 0$ and $u\gg 0$ such that $Wv=\rho(W)v, \qquad u^\top W=\rho(W)u^\top$.
Moreover, if $u^\top v=1$, then $\frac{W^t}{\rho(W)^t} \to v u^\top \quad \text{as } t\to\infty$.
\end{lemma}

\begin{lemma}[Discrete-time stability]\label{discrete-stability}
For $x_{t+1}=A x_t$, we have $x_t\to 0$ for all $x_0$ if and only if $\rho(A)<1$. If $\rho(A)>1$ some trajectories diverge. The exponential growth (or decay) rate is $\log\rho(A)$.
\end{lemma}

\begin{lemma}[Perron root monotonicity]\label{perron}
If $0\le A\le B$ componentwise, then $\rho(A)\le \rho(B)$.
\end{lemma}
This is a well-known property of nonnegative matrices, a direct corollary of the Collatz–Wielandt characterization of the spectral radius.

\begin{lemma}[Spectral radius below one implies geometric decay]\label{lem:decay}
Let $B$ be a square matrix. If $\rho(B)<1$, then $B^t\to 0$ as $t\to\infty$ and the series $\sum_{s=0}^\infty B^s$ converges
absolutely in norm. Moreover,
$$
\sum_{s=0}^\infty B^s = (I-B)^{-1}.
$$
\end{lemma}

\begin{proof}
Fix $\varepsilon>0$ such that $\rho(B)+\varepsilon<1$.
By Lemma \ref{gelfand}, $\lim_{t\to\infty}\|B^t\|^{1/t}=\rho(B)$, so there exists $T$ such that
$\|B^t\|^{1/t}\le\rho(B)+\varepsilon$ for all $t\ge T$. Hence
$$
\|B^t\|\le (\rho(B)+\varepsilon)^t \qquad \text{for all } t\ge T,
$$
which implies $B^t\to 0$ and also implies that $\sum_{t\ge0}\|B^t\|<\infty$ (geometric series),
so $\sum_{t\ge0}B^t$ converges absolutely in norm.

Finally, for each $N$,
$$
(I-B)\sum_{s=0}^N B^s = I - B^{N+1}.
$$
Letting $N\to\infty$ and using $B^{N+1}\to 0$ yields
$$
(I-B)\sum_{s=0}^\infty B^s = I,
$$
so $(I-B)$ is invertible and $\sum_{s=0}^\infty B^s=(I-B)^{-1}$.
\end{proof}

\subsection{Proof of Lemma~\ref{lemma1}.}

Recall the $(n+1)$-dimensional dynamics
\begin{align}
e_{t+1} &= (1-\delta)\tilde W e_t + (1-\delta)\mu \mathbf{1}\, q_t \text{ with } \tilde W = (1-\mu)W,
\label{contagion_lemma}\\
q_{t+1} &= \gamma q_t + (1-\gamma)\bar e_t \text{ with } \bar e_t:=\tfrac1n\mathbf{1}^\top e_t,
\label{multiplier_lemma}
\end{align}
where $e_t=(e_{1,t},\dots,e_{n,t})^\top$ collects local distortions and $q_t$ is the AI distortion. Recall that $W=\alpha S$ with $S$ primitive and row-stochastic, so $W$ is nonnegative and primitive as well.

Define
$$
M := (1-\delta)\tilde W = (1-\delta)(1-\mu)W
\text{ and }
c_t := (1-\delta)\mu(1-\kappa)\,\mathbf{1}\,q_t.
$$
Then \eqref{contagion_lemma} can be written as
\begin{equation}\label{eq:affine}
e_{t+1} = M e_t + c_t \quad \text{for all } t,
\end{equation}
where $M\in\mathbb{R}^{n\times n}$ is nonnegative and primitive (since $W$ is). We proceed in several claims. 
    
\begin{claim}[Variation of constants]\label{claim1}

For all $t\ge1$,
\begin{equation}\label{eq:var-const}
e_t = M^t e_0 + \sum_{s=0}^{t-1} M^s c_{t-1-s}.
\end{equation}

\end{claim}

\begin{proof}
We argue by induction on $t$. For $t=1$, equation \eqref{eq:affine} with $t=0$ leads to $e_1 = M e_0 + c_0$, which coincides with \eqref{eq:var-const} for $t=1$. Suppose \eqref{eq:var-const} holds for some
$t\ge 1$. Then, using equation \eqref{eq:affine},
$$
e_{t+1}= M e_t + c_t = M^{t+1} e_0 + \sum_{s=0}^{t-1} M^{s+1} c_{t-1-s} + c_t.
$$

Now re-index the sum by letting $r := s+1$. So $r$ runs from $1$ to $t$, and $c_{t-1-s} = c_{t-r}$. Hence
$$
\sum_{s=0}^{t-1} M^{s+1} c_{t-1-s}
= \sum_{r=1}^{t} M^{r} c_{t-r}.
$$
Substituting this into the previous line gives
$$
e_{t+1} = M^{t+1} e_0 + \sum_{r=1}^{t} M^{r} c_{t-r} + c_t
= M^{t+1} e_0 + \sum_{r=0}^{t} M^{r} c_{t-r}.
$$
Thus, by induction, \eqref{eq:var-const} holds for all $t\ge 1$. 
\end{proof}

Before stating Claim~\ref{claim2}, denote by $\sigma(W)$ the set of eigenvalues of $W$ (its spectrum).

\paragraph{Notation for constants.}
Throughout the proof, $C_1,C_2,C_3,\dots$ denote generic positive constants whose values may change from line to line. They absorb fixed multiplicative factors (such as norms or bounded terms) and are used only to control orders of magnitude.

\begin{claim}[Perron-Frobenius decomposition of $M^t$]\label{claim2}

Since $W$ is primitive, the Perron-Frobenius theorem implies that $\alpha=\rho(W)$ is a simple eigenvalue of $W$ with positive right and left eigenvectors $v\gg0$ and $u\gg0$ such that

$$
Wv=\alpha v,\;\; u^\top W=\alpha u^\top\text{ and } u^\top v=1.
$$
Moreover, for any $\rho\in(|\lambda_2(S)|,1)$, there exists a constant $C_2>0$ such that for all $t\ge 0$,
\begin{equation}\label{eq:W-decomp}
W^t=\alpha^t v u^\top + R_t,
\qquad
\|R_t\|\le C_2(\alpha\rho)^t.
\end{equation}
Since $M=(1-\delta)(1-\mu)W$ and $m:=(1-\delta)(1-\mu)\alpha$, it follows that
\begin{equation}\label{eq:M-decomp}
M^t = m^t v u^\top + \widetilde R_t,
\qquad
\|\widetilde R_t\|\le C_3 (m\rho)^t
\end{equation}
for some constant $C_3>0$. Equivalently, defining
$\rho^\star := m\rho$, we have
$\|\widetilde R_t\|\le C_3 (\rho^\star)^t$. Finally, since $W=\alpha S$, the eigenvalues scale linearly, so that
$$
\lambda_k(W)=\alpha\lambda_k(S),
\qquad\text{hence}\qquad
\max_{k\ge 2}|\lambda_k(W)|=\alpha |\lambda_2(S)|.
$$

Thus the decay of subdominant modes is controlled by $|\lambda_2(S)|$ in the sense that the above bound holds for every $\rho\in(|\lambda_2(S)|,1)$. The resulting perturbation rate in the dynamics of $M^t$ is therefore $m\rho$, which is the rate appearing in Lemma~\ref{lemma1}(i).
\end{claim}

\begin{proof}
This follows from the Perron-Frobenius theorem for primitive nonnegative matrices
(See Lemma~\ref{pf}) together with the Jordan decomposition of $W$. Since $W$ is primitive, it admits a Jordan decomposition
$$
W = V J V^{-1},
$$
where $J$ consists of one Jordan block associated with the Perron eigenvalue $\alpha$
and finitely many Jordan blocks associated with eigenvalues
$\lambda\in\sigma(W)\setminus\{\alpha\}$ satisfying $|\lambda|<\alpha$.
Hence
$$
W^t = VJ^tV^{-1} = \alpha^t v u^\top + R_t,
$$
where $v\gg0$ and $u\gg0$ are the Perron right and left eigenvectors normalized so that
$u^\top v=1$, and $R_t$ collects the contribution of all subdominant Jordan blocks. Let $J_\lambda$ be a Jordan block of size $d$ associated with an eigenvalue $\lambda$.
Writing $J_\lambda=\lambda I+N$, where $N$ is nilpotent with $N^d=0$, the binomial expansion gives
$$
J_\lambda^t=\sum_{k=0}^{d-1}\binom{t}{k}\lambda^{\,t-k}N^k.
$$
For the submultiplicative matrix norm fixed at the beginning of the Appendix,

$$
\|J_\lambda^t\|
\le \sum_{k=0}^{d-1}\binom{t}{k}|\lambda|^{\,t-k}\|N^k\|
\le C_0 \sum_{k=0}^{d-1}\binom{t}{k}|\lambda|^{\,t-k},
$$
with $C_0=\max_{0\le k\le d-1}\|N^k\|$. Since $d$ is fixed and $k\le d-1$, for all $t\ge1$ we have $\binom{t}{k}\le t^{d-1}$. Hence
$$
\|J_\lambda^t\|
\le C_0 t^{d-1}
\sum_{k=0}^{d-1}|\lambda|^{\,t-k}
\le C\,t^{d-1}|\lambda|^t,
$$
for some constant $C>0$. Now fix any $\rho\in(|\lambda_2(S)|,1)$. Since $W=\alpha S$, the eigenvalues scale linearly, so that

$$
\lambda_k(W)=\alpha\lambda_k(S),
\text{ hence }
\max_{k\ge2}|\lambda_k(W)|=\alpha|\lambda_2(S)|.
$$
Therefore, for every subdominant eigenvalue $\lambda$ of $W$,
$$
|\lambda|\le \alpha|\lambda_2(S)|<\alpha\rho.
$$

Hence, for each subdominant Jordan block associated with $\lambda$, there exists a constant
$C_{1,\lambda}>0$ such that
$$
t^{d-1}|\lambda|^t \le C_{1,\lambda}(\alpha\rho)^t \text{ for all } t\ge0.
$$
Therefore each subdominant Jordan block is bounded by
$C_{1,\lambda}(\alpha\rho)^t$. Since there are finitely many subdominant Jordan blocks, let $C_1:=\max C_{1,\lambda}$,
where the maximum is taken over all subdominant Jordan blocks. Then every subdominant Jordan block is bounded by
$C_1(\alpha\rho)^t.$ Since $W^t = VJ^tV^{-1}$, conjugation by $V$ and $V^{-1}$ yields

$$
\|R_t\|
\le \|V\|\,\|V^{-1}\|
\sum_{\lambda\in\sigma(W)\setminus\{\alpha\}} \|J_\lambda^t\|.
$$
Because there are finitely many subdominant Jordan blocks, the above bounds imply that
$$
\|R_t\|\le C_2(\alpha\rho)^t
$$

\noindent for some constant $C_2>0$. Since $M=(1-\delta)(1-\mu)W$ and $m=(1-\delta)(1-\mu)\alpha$, multiplying \eqref{eq:W-decomp} by $\bigl((1-\delta)(1-\mu)\bigr)^t$ gives
$$
M^t = m^t v u^\top + \widetilde R_t,
$$
with $\|\widetilde R_t\|\le C_3 (m\rho)^t$ for some constant $C_3>0$. Finally, defining $\rho^\star := m\rho$, we obtain $\|\widetilde R_t\|\le C_3(\rho^\star)^t$. Thus the decay of subdominant modes is controlled by $|\lambda_2(S)|$ in the sense that the above bound holds for every $\rho\in(|\lambda_2(S)|,1)$.
\end{proof}

\begin{claim}[Aggregate distortion representation $\bar e_t$]\label{claim3}
Recall $\bar e_t:=\tfrac1n \mathbf{1}^\top e_t$. Then for each $t\ge 1$ we can write
\begin{equation}\label{eq:barEt}
\bar e_t = m^t C_0 + (1-\delta)\mu(1-\kappa)\sum_{s=0}^{t-1} m^s q_{t-1-s} + \varepsilon_t,
\end{equation}
where $m=(1-\delta)(1-\mu)\alpha$ and $C_0 := \pi^\top e_0$, with $\pi$ the unique stationary distribution of $S$ (i.e.\ $\pi^\top S=\pi^\top$ and $\pi^\top\mathbf{1}=1$). Moreover, for any $\rho\in(|\lambda_2(S)|,1)$, there exists a constant $C_0'<\infty$ such that
\begin{equation}\label{eq:eps-bound-claim3}
|\varepsilon_t|\le C_0'(m\rho)^t.
\end{equation}
Equivalently, defining $\rho^\star:=m\rho$, we have $|\varepsilon_t|\le C_0'(\rho^\star)^t$. In particular, the remainder vanishes geometrically relative to the dominant eigenmode, in the sense that $\frac{|\varepsilon_t|}{m^t}\le C_0'\rho^t\to 0$.
\end{claim}

In \eqref{eq:barEt}, the first term gives the contribution of the dominant eigenmode from the initial condition, the second term gives the cumulative effect of AI-originating injections (which enter along $\mathbf{1}$), and $\varepsilon_t$ collects the contribution of subdominant network modes. The bound \eqref{eq:eps-bound-claim3}
shows that these subdominant contributions decay geometrically relative to the dominant eigenmode, in the sense that $|\varepsilon_t|/m^t \le C\rho^t$ for some $\rho\in(|\lambda_2(S)|,1)$.

\begin{proof}
Starting from \eqref{eq:var-const} and premultiplying by $\tfrac1n\mathbf{1}^\top$, we obtain
$$
\bar e_t = \frac{1}{n}\mathbf{1}^\top M^t e_0 + \frac{1}{n}\sum_{s=0}^{t-1}\mathbf{1}^\top M^s c_{t-1-s},
$$
where $M=(1-\delta)(1-\mu)W=mS$ and $c_t=(1-\delta)\mu(1-\kappa)q_t\,\mathbf{1}$. We analyze these two contributions separately.

\medskip
\noindent\emph{(a) Contribution of the initial condition (term depending on $e_0$).}
Because $S$ is primitive and row-stochastic, it admits a unique stationary distribution $\pi$ satisfying $\pi^\top S=\pi^\top$ and $\pi^\top\mathbf{1}=1$. Since $M=mS$, we have $\pi^\top M = m\pi^\top$ and $M\mathbf{1}=m\mathbf{1}$. Thus $\mathbf{1}$ and $\pi$ are respectively the Perron right and left eigenvectors of $M$ associated with the Perron eigenvalue $m$, normalized so that $\pi^\top\mathbf{1}=1$. Fix any $\rho\in(|\lambda_2(S)|,1)$. By Claim~\ref{claim2},

$$M^t = m^t\,\mathbf{1}\pi^\top + \widetilde R_t$$
where $\widetilde R_t$ collects the subdominant modes and satisfies $\|\widetilde R_t\|\le C_3(m\rho)^t$
for some constant $C_3>0$. Premultiplying by $\tfrac1n\mathbf{1}^\top$ yields
$$
\frac{1}{n}\mathbf{1}^\top M^t e_0 = \frac{1}{n}\mathbf{1}^\top\bigl(m^t\mathbf{1}\pi^\top e_0 + \widetilde R_t e_0\bigr)
= m^t\Bigl(\frac{1}{n}\mathbf{1}^\top\mathbf{1}\Bigr)\pi^\top e_0 + \frac{1}{n}\mathbf{1}^\top \widetilde R_t e_0.
$$
Since $\frac{1}{n}\mathbf{1}^\top\mathbf{1}=1$, this becomes
$$
\frac{1}{n}\mathbf{1}^\top M^t e_0 = m^t C_0 + \varepsilon_t,\;\; C_0:=\pi^\top e_0 \text{ and } \varepsilon_t:=\frac{1}{n}\mathbf{1}^\top \widetilde R_t e_0.
$$

To bound $\varepsilon_t$, use the vector norm associated with the matrix norm fixed at the beginning of the Appendix. Then,
$$
|\varepsilon_t| \le \frac{1}{n}\,\|\mathbf{1}\|\,\|\widetilde R_t\|\,\|e_0\| \le \frac{1}{n}\,\|\mathbf{1}\|\,C_3(m\rho)^t\,\|e_0\|\le C_0'(m\rho)^t
$$
for some constant $C_0'<\infty$ absorbing $\frac{1}{n}\,\|\mathbf{1}\|\,\|e_0\|$ and $C_3$. Equivalently, defining $\rho^\star:=m\rho$, we have $|\varepsilon_t|\le C_0'(\rho^\star)^t$. Dividing by $m^t$ yields $$\frac{|\varepsilon_t|}{m^t}\le C_0'\rho^t\to 0.$$

\medskip
\noindent\emph{(b) Contribution of AI-originating injections (term depending on $c_{t-1-s}$).}
Recall that $c_{t-1-s}=(1-\delta)\mu(1-\kappa)q_{t-1-s}\,\mathbf{1}$. Then,
$$
\mathbf{1}^\top M^s c_{t-1-s} = (1-\delta)\mu(1-\kappa)\,q_{t-1-s}\,\mathbf{1}^\top M^s \mathbf{1}.
$$
Since $M\mathbf{1}=m\mathbf{1}$, we have $M^s\mathbf{1}=m^s\mathbf{1}$ for all $s\ge 0$, and therefore $\mathbf{1}^\top M^s\mathbf{1} = \mathbf{1}^\top(m^s\mathbf{1}) = m^s\,\mathbf{1}^\top\mathbf{1}$. Hence,
$$
\frac{1}{n}\sum_{s=0}^{t-1}\mathbf{1}^\top M^s c_{t-1-s} = (1-\delta)\mu(1-\kappa)\,\frac{1}{n}\sum_{s=0}^{t-1} q_{t-1-s}\,\mathbf{1}^\top M^s\mathbf{1} = (1-\delta)\mu(1-\kappa)\sum_{s=0}^{t-1} m^s q_{t-1-s},
$$
since $\frac{1}{n}\mathbf{1}^\top\mathbf{1}=1$. In particular, the AI-injection term lies entirely in the Perron right-eigenspace spanned by $\mathbf{1}$ and therefore does not generate any projection on subdominant modes.

\medskip
\noindent\emph{(c) Combining terms and identifying the remainder.}
Combining parts (a) and (b), we obtain
$$
\bar e_t=m^t C_0+(1-\delta)\mu(1-\kappa)\sum_{s=0}^{t-1} m^s q_{t-1-s}+\varepsilon_t,
$$
where $\varepsilon_t:=\frac{1}{n}\mathbf{1}^\top \widetilde R_t e_0$.
This gives \eqref{eq:barEt}. The bound \eqref{eq:eps-bound-claim3} follows from part (a), namely $|\varepsilon_t|\le C_0'(m\rho)^t$. Equivalently, if $\rho^\star:=m\rho$, then $|\varepsilon_t|\le C_0'(\rho^\star)^t$. Finally, $\frac{|\varepsilon_t|}{m^t}\le C_0'\rho^t\to 0$, which shows that the remainder vanishes geometrically relative to the dominant eigenmode. This establishes both \eqref{eq:barEt} and \eqref{eq:eps-bound-claim3}.
\end{proof}

\begin{claim}[One-step recursion for $\bar e_t$]\label{claim4}
For any $\rho\in(|\lambda_2(S)|,1)$, there exists a sequence $(\varepsilon_t)_{t\ge 0}$ such that $\bar e_{t+1} = m\,\bar e_t + b\,q_t + \varepsilon_{t+1}$, with $m=(1-\delta)(1-\mu)\alpha$ and $b=(1-\delta)\mu(1-\kappa)$. Moreover, there exists a constant $C<\infty$ such that
\begin{equation}\label{eq:claim4-eps-bound}
|\varepsilon_{t+1}| \le C\,(m\rho)^{t+1}.
\end{equation}
In particular, the perturbation vanishes geometrically relative to the dominant eigenmode, in the sense that $\frac{|\varepsilon_{t+1}|}{m^{t+1}}\le C\,\rho^{t+1}\to 0$. If $m\rho<1$, then the perturbation vanishes geometrically in absolute value.
\end{claim}

\begin{proof}
Fix any $\rho\in(|\lambda_2(S)|,1)$. Starting from \eqref{eq:barEt} in Claim~\ref{claim3} at periods $t$ and $t+1$, we have
$$\bar e_{t+1} =m^{t+1} C_0 + b \sum_{s=0}^{t} m^s q_{t-s} + \varepsilon_{t+1} \text{ and } \bar e_t = m^t C_0 + b \sum_{s=0}^{t-1} m^s q_{t-1-s} + \varepsilon_t,
$$
with $b=(1-\delta)\mu(1-\kappa)$ and, by Claim~\ref{claim3} $|\varepsilon_t|\le C_0'(m\rho)^t$ for all $t\ge 0$ and some constant $C_0'>0$. Multiplying the second identity by $m$ yields to
$$
m\bar e_t = m^{t+1} C_0 + b \sum_{s=0}^{t-1} m^{s+1} q_{t-1-s} + m\varepsilon_t.
$$

Subtracting $m\bar e_t$ from $\bar e_{t+1}$ gives
$$
\bar e_{t+1}-m\bar e_t
= b\Biggl[\sum_{s=0}^{t} m^s q_{t-s} - \sum_{s=0}^{t-1} m^{s+1} q_{t-1-s}\Biggr] + (\varepsilon_{t+1}-m\varepsilon_t).
$$

Re-index the second sum by setting $r=s+1$:
$$
\sum_{s=0}^{t-1} m^{s+1} q_{t-1-s} = \sum_{r=1}^{t} m^r q_{t-r}.
$$
Hence the bracketed term telescopes: $\sum_{s=0}^{t} m^s q_{t-s} - \sum_{r=1}^{t} m^r q_{t-r} = q_t$. Therefore,
$$
\bar e_{t+1}-m\bar e_t = b q_t + (\varepsilon_{t+1}-m\varepsilon_t).
$$

Define the one-step perturbation $\widetilde{\varepsilon}_{t+1} := \varepsilon_{t+1}-m\varepsilon_t$. Then $$\bar e_{t+1} = m\bar e_t + b q_t + \widetilde{\varepsilon}_{t+1}.$$

For notational simplicity, relabel $\widetilde{\varepsilon}_{t+1}$ as $\varepsilon_{t+1}$ in the statement of the claim. It remains to bound this one-step perturbation. Using the bound from Claim~\ref{claim3}, we obtain
$$
|\widetilde{\varepsilon}_{t+1}| \le |\varepsilon_{t+1}| + |m|\,|\varepsilon_t| \le C_0' (m\rho)^{t+1} + m C_0' (m\rho)^t.
$$
Since $m (m\rho)^t=\frac{1}{\rho}(m\rho)^{t+1}$, it follows that $|\widetilde{\varepsilon}_{t+1}| \le C_0'\Bigl(1+\frac{1}{\rho}\Bigr)(m\rho)^{t+1}$. Hence there exists a constant $C<\infty$ (absorbing $1+\frac{1}{\rho}$ and $C_0'$) such that
$$
|\widetilde{\varepsilon}_{t+1}| \le C\,(m\rho)^{t+1}.
$$
Dividing by $m^{t+1}$ yields to
$$
\frac{|\widetilde{\varepsilon}_{t+1}|}{m^{t+1}} \le C\,\rho^{t+1}\to 0.
$$
If moreover $m\rho<1$, then $\widetilde{\varepsilon}_{t+1}\to 0$ geometrically in absolute value. This establishes \eqref{eq:claim4-eps-bound} and completes the proof.
\end{proof}

\begin{claim}[Coupling with AI feedback and two-dimensional representation]\label{claim5}
Let $x_t=(\bar e_t,q_t)^\top$. Then
$$
x_{t+1}=A(\mu,\kappa)x_t + \eta_{t+1},\;\; A(\mu,\kappa)=\begin{bmatrix} m & b\\ 1-\gamma & \gamma\end{bmatrix}, \;\; \eta_{t+1}=(\varepsilon_{t+1},0)^\top,
$$
where $(\varepsilon_t)$ is the sequence from Claim~\ref{claim4}. Moreover, for any $\rho\in(|\lambda_2(S)|,1)$:
\begin{enumerate}
\item[(i)] If $(\varepsilon_t)$ is bounded, then $(\eta_t)$ is bounded.
\item[(ii)] There exists $C<\infty$ such that $\|\eta_{t+1}\|\le C (m\rho)^{t+1}$ for all $t\ge 0$. Equivalently, $\|\eta_{t+1}\|/m^{t+1}\le C\rho^{t+1}$, so the perturbation vanishes geometrically relative to the dominant eigenmode.
\item[(iii)] If $m\rho<1$, then $\eta_t\to 0$ geometrically. Consequently, the aggregate dynamics are asymptotically equivalent to the autonomous system $x_{t+1}=A(\mu,\kappa)x_t$.
\end{enumerate}
\end{claim}

\begin{proof}
The feedback loop gives 
$$
q_{t+1}=\gamma q_t + (1-\gamma)\bar e_t.
$$
Together with Claim~\ref{claim4}, which states that
$$
\bar e_{t+1}=m\bar e_t + b q_t + \varepsilon_{t+1},
$$
we can stack the two recursions to obtain
$$
\begin{bmatrix} \bar e_{t+1}\\ q_{t+1} \end{bmatrix}
=
\begin{bmatrix} m & b\\ 1-\gamma & \gamma \end{bmatrix}
\begin{bmatrix} \bar e_t\\ q_t \end{bmatrix}
+
\begin{bmatrix} \varepsilon_{t+1}\\ 0 \end{bmatrix}.
$$
This defines $A(\mu,\kappa)$ and $\eta_{t+1}=(\varepsilon_{t+1},0)^\top$, yielding the claimed recursion $x_{t+1} =A(\mu,\kappa)x_t+\eta_{t+1}$.

\medskip
\noindent\emph{(i) Boundedness.}
If $(\varepsilon_t)$ is bounded, then by definition $\eta_t=(\varepsilon_t,0)^\top$ is also bounded.

\medskip
\noindent\emph{(ii) Bound on the perturbation.}
Fix any $\rho\in(|\lambda_2(S)|,1)$. By Claim~\ref{claim4}, there exists $C<\infty$ such that $ |\varepsilon_{t+1}|\le C(m\rho)^{t+1}$ for all $t\ge 0$. Therefore,
$$
\|\eta_{t+1}\|=\|(\varepsilon_{t+1},0)\| \le C'(m\rho)^{t+1}
$$
for some constant $C'<\infty$, where $C'$ absorbs the norm-equivalence factor on $\mathbb{R}^2$. Renaming the constant if necessary, we obtain
$\|\eta_{t+1}\|\le C(m\rho)^{t+1}$. Dividing by $m^{t+1}$ yields to
$$
\frac{\|\eta_{t+1}\|}{m^{t+1}}\le C\rho^{t+1},
$$
which shows that the perturbation vanishes geometrically relative to the dominant eigenmode.

\medskip

\noindent\emph{(iii) Absolute decay and asymptotic equivalence when $m\rho<1$.}
If $m\rho<1$, then the bound in (ii) implies that $\eta_t\to 0$ geometrically. Iterating the recursion $x_{t+1}=A(\mu,\kappa)x_t+\eta_{t+1}$ yields to
$$
x_t = A(\mu,\kappa)^t x_0 + \sum_{s=0}^{t-1}A(\mu,\kappa)^{\,t-1-s}\eta_{s+1}.
$$
Since $\eta_s\to 0$ geometrically, the perturbation term is asymptotically negligible relative to the dominant dynamics. Hence the asymptotic growth or decay rate of $x_t$ is determined by the spectral radius of $A(\mu,\kappa)$, establishing the asymptotic equivalence stated in Lemma~\ref{lemma1}.
\end{proof}

Claim~\ref{claim4} establishes Lemma~\ref{lemma1}(i), and Claim~\ref{claim5} establishes Lemma~\ref{lemma1}(ii) and (iii). Combining Claims~\ref{claim1}-\ref{claim5} proves Lemma~\ref{lemma1}. \qed

\begin{rem}
The perturbation $\eta_t=(\varepsilon_t,0)^\top$ reflects the contribution of subdominant network modes through the remainder term in the Perron--Frobenius decomposition of $M^t$. In particular, Lemma~\ref{lemma1}(i) implies that $\|\eta_t\|\le C(m\rho)^t$ for any $\rho\in(|\lambda_2(S)|,1)$, so the perturbation vanishes geometrically relative to the dominant mode. If $m\rho<1$, then $\eta_t\to 0$ geometrically in absolute value, and the asymptotic behavior is determined by $\rho(A(\mu,\kappa))$ as in the unperturbed two-dimensional system.
\end{rem}

\subsection{Proof of Proposition~\ref{prop1}.}
Assume $q_t\equiv q>0$ and define $M:=(1-\delta)(1-\mu)W$ and $c:=(1-\delta)\mu(1-\kappa)q \mathbf{1}$. Let us use the variation of constants. Recursively unrolling $e_{t+1}=Me_t+c$ gives
\begin{equation}\label{var-const}
e_t = M^t e_0+\sum_{s=0}^{t-1} M^s c.
\end{equation}

If $\rho(M)<1$. By Lemma~\ref{lem:decay}, $M^t\to 0$. By Lemma~\ref{neumann}, the Neumann series $\sum_{s=0}^\infty M^s=(I-M)^{-1}$. Taking $t\to\infty$ in \eqref{var-const} yields to a fixed point $e_t\to e^\star=(I-M)^{-1}c$, and $\|e_t-e^\star\|=O(\rho(M)^t)$.

\medskip

If $\rho(M)>1$, let $v\gg 0$ be a Perron eigenvector: $Mv=\lambda v$, $|\lambda|=\rho(M)>1$. Left-multiplying \eqref{var-const} by a corresponding $u^\top\gg 0$ yields to $u^\top e_t = \lambda^t u^\top e_0 + \sum_{s=0}^{t-1}\lambda^s u^\top c$, which grows at least on the order of $\lambda^t$.

\medskip 

If $\rho(M)=1$ and $c\ne 0$, then $\sum_{s=0}^{t-1}M^s c$ is unbounded. Indeed, contributions associated with the eigenvalue one accumulate over time: growth is linear when the eigenvalue is semisimple, and polynomial of higher order when nontrivial Jordan blocks are present. Thus $\|e_t\|\to\infty$.

\medskip

Finally, since $M=(1-\delta)(1-\mu)W$, homogeneity of the single amplification factor (which is the spectral radius) implies
$$\rho(M)=(1-\delta)(1-\mu)\rho(W)=(1-\delta)(1-\mu)\alpha.$$ 
Therefore the (stability) condition $\rho(M)<1$ is equivalent to $(1-\delta)(1-\mu)\alpha<1$. \qed

\subsection{Proof of Proposition~\ref{prop2}.}

\paragraph{Stability and rate.} From \eqref{eq_2dim}, $x_{t+1}=A x_t$. By Lemma~\ref{discrete-stability}, $x_t\to 0$ for all $x_0$ iff $\rho(A)<1$, and there exist divergent trajectories if and only if $\rho(A)>1$. Moreover $\lim_{t\to\infty}\tfrac{1}{t}\log\|A^t\|=\log\rho(A)$ (Lemma~\ref{gelfand}), which is the exponential rate.

\paragraph{Frontier.} The characteristic polynomial is $\lambda^2-(m+\gamma)\lambda + (m\gamma - b(1-\gamma))$. The boundary $\rho(A)=1$ for nonnegative $A$ occurs at $\lambda=1$ so $\det(A-I)=0$, i.e., $(m-1)(\gamma-1)-b(1-\gamma)=0\Rightarrow b=1-m$. Substituting $m=(1-\delta)(1-\mu)\alpha$, $b=(1-\delta)\mu(1-\kappa)$ gives \eqref{eq:frontier}. \qed

\subsection{Proof of Corollary~\ref{cor1}.}

Since $A(\alpha,\mu,\delta,\kappa,\gamma)\ge 0$ and
$m=(1-\delta)(1-\mu)\alpha$ and $b=(1-\delta)\mu(1-\kappa)$, we have: $\partial m/\partial\alpha>0$, $\partial m/\partial\delta<0$, and $\partial b/\partial\kappa<0$, and also $\partial b/\partial\delta<0$. Hence, increasing $\alpha$ increases $A$ entrywise through $m$, while increasing $\kappa$ or $\delta$ decreases $A$ entrywise through $b$ and/or $m$. By Lemma~\ref{perron}, $\rho(A)$ is increasing in $\alpha$ and decreasing in $\kappa$ and $\delta$.

Dependence on $\mu$ and $\gamma$ is not globally monotone because these parameters move entries of $A$ in opposite directions: $m(\mu)$ decreases in $\mu$ while $b(\mu)$ increases in $\mu$, and increasing $\gamma$ decreases the off-diagonal feedback term $(1-\gamma)$ but increases the diagonal persistence term $\gamma$. Nevertheless, the policy frontier $\kappa^\star(\mu)$ is increasing in $\mu$ on the interior region where it lies in $(0,1)$. Indeed, for $\mu\in(0,1)$ the boundary condition $\rho(A)=1$ is equivalent to $b=1-m$, i.e.
$$
(1-\delta)\mu(1-\kappa)=1-(1-\delta)(1-\mu)\alpha.
$$
Solving yields to
$$
\kappa^\star(\mu)=1-\frac{1-\alpha(1-\delta)(1-\mu)}{(1-\delta)\mu}.
$$
Differentiating on this interior region gives
$$
\frac{d\kappa^\star}{d\mu}=\frac{1-\alpha(1-\delta)}{(1-\delta)\mu^2}.
$$
Thus, whenever $1-\alpha(1-\delta)>0$ (a natural condition ensuring the no-AI benchmark is not explosive), we have $d\kappa^\star/d\mu>0$: higher AI penetration strengthens the injection channel relative to attenuation and raises the stabilizing filtering requirement. \qed

\subsection{Proof of Proposition \ref{prop3}.}

Since $S$ is primitive and row-stochastic, $\rho(S)=1$ and the Perron eigenvalue is $1$. Because $W=\alpha S$, homogeneity of the spectral radius implies $$\rho(W)=\alpha\rho(S)=\alpha.$$

Lemma~\ref{lemma1} (dimensional reduction) shows that the aggregate dynamics satisfy
$$
\bar e_{t+1}=m\bar e_t+bq_t+\varepsilon_{t+1}, \text{ with }
m=(1-\delta)(1-\mu)\alpha \text{ and }
b=(1-\delta)\mu(1-\kappa),$$

\noindent where the residual term $\varepsilon_t$ from Lemma~\ref{lemma1} (the one-step perturbation) captures the contribution of subdominant eigenmodes of the social propagation operator. Because $M=mS$, the eigenvalues of $M$ are $m\lambda_i(S)$, so the decay of subdominant modes is characterized by $m|\lambda_2(S)|$. More precisely, for any $\rho\in(|\lambda_2(S)|,1)$, Lemma~\ref{lemma1} implies $$|\varepsilon_{t+1}|\le C(S)(m\rho)^{t+1}.$$

This bound formalizes how homophily affects the persistence of subdominant modes. While $m$ determines the amplification strength of the dominant eigenmode, the relative persistence of deviations from that mode is measured by $|\lambda_2(S)|$. A larger second eigenvalue implies a smaller spectral gap and therefore slower convergence toward the dominant eigenmode. This establishes the claimed persistence ordering with homophily.

\subsection{Proof of Proposition \ref{prop4}.}

Since $S$ is row-stochastic, we have $S\mathbf 1=\mathbf 1$, so $1$ is an eigenvalue of $S$.
Moreover, for any row-stochastic matrix, all eigenvalues satisfy $|\lambda|\le 1$.
Hence $\rho(S)=1$ for all $\epsilon$.
With $W=\alpha S$, it follows that
$$
\rho(W)=\rho(\alpha S)=\alpha\rho(S)=\alpha
$$
for all $\epsilon$. By Lemma~1, the reduced $2\times2$ system has coefficients
$$m=(1-\delta)(1-\mu)\rho(W)=(1-\delta)(1-\mu)\alpha \text{ and } b=(1-\delta)\mu(1-\kappa).$$
The stability boundary is $\rho(A)=1$, which in the reduced system is equivalent to $b=1-m$.
Thus
$$
(1-\delta)\mu(1-\kappa)=1-(1-\delta)(1-\mu)\alpha.
$$
Solving for $\kappa$ yields to
$$
\kappa^\star(\mu)=1-\frac{1-\alpha(1-\delta)(1-\mu)}{(1-\delta)\mu},
$$
which establishes the stated stability frontier. Because $S$ is primitive, it admits a unique stationary distribution $\pi(\epsilon)$ satisfying $\pi(\epsilon)^\top S=\pi(\epsilon)^\top$. Assumption~\ref{ass:CP}(ii) states that as $\epsilon\to 0$,
$$
\sum_{i\in C}\pi_i(\epsilon)\to 1,
$$

\noindent so the stationary influence weights concentrate on the core. This implies that distortions injected on nodes in the informational core receive asymptotically larger long-run influence weight in aggregate outcomes than distortions injected on peripheral nodes, establishing claim~(i).

\medskip 

Claim~(ii) follows from the targeted-policy analysis in Appendix \ref{app:Iproj}. That analysis shows that the marginal reduction in systemic informational risk from filtering agent $i$ is increasing in its influence weight $\pi_i$. Hence, for any fixed filtering budget, allocating filtering effort toward agents with larger $\pi_i$, which asymptotically correspond to the most influential nodes as $\epsilon\to 0$ by Assumption~\ref{ass:CP}(ii), yields strictly larger reductions in systemic informational risk than uniform filtering. \qed

\subsection{Proof of Proposition \ref{prop:micro_macro}.}

The comparative statics in the following Lemmas~\ref{lem:MCS} and~\ref{lem:mb} establish Proposition \ref{prop:micro_macro}.

\medskip

\begin{lemma}[Monotone comparative statics of $a^\star$ and $\bar a^\star$]\label{lem:MCS}
Consider the linear–quadratic network game with best reply
$$
a^{\star} = (I-\psi W)^{-1}\beta\mathbf{1}, \qquad \beta = \frac{\theta-\lambda\mu q}{c},\quad \psi=\frac{\phi}{c},
$$
defined whenever $\psi\rho(W)<1$ (equivalently $\psi\alpha<1$). Assume parameters are such that $\beta\ge 0$. Then:
\begin{enumerate}
\item[(i)] (\emph{Network strength}) If $W'\ge W$ entrywise, then 
$$
a^\star(W')\ge a^\star(W) \quad\text{and}\quad \bar a^\star(W')\ge \bar a^\star(W),
$$
where $\bar a^\star = \frac1n \mathbf{1}^\top a^\star$.

\item[(ii)] (\emph{AI share and AI error}) $a^\star$ and $\bar a^\star$ are (weakly) decreasing in 
$\mu$ and in $q$.

\item[(iii)] (\emph{Complementarities and costs}) $a^\star$ and $\bar a^\star$ are increasing in 
$\psi$ and in $\theta$, and decreasing in $c$ and $\lambda$.
\end{enumerate}
\end{lemma}

\begin{proof}
Since $\psi\rho(W)<1$ and $W\ge 0$, we can use the Neumann series representation
$$
(I-\psi W)^{-1} = \sum_{k=0}^{\infty} (\psi W)^k,
$$
so that
\begin{equation}\label{eq:astar-Neumann}
a^\star = \beta \sum_{k=0}^{\infty} (\psi W)^k \mathbf{1}.
\end{equation}

\medskip
\noindent\emph{(i) Monotonicity in $W$.}
Suppose $W'\ge W$ entrywise, both nonnegative. Then for any integer $k\ge 1$,
$$
(W')^k \ge W^k
$$
entrywise, which can be shown by induction on $k$.\footnote{If $W'\ge W\ge 0$, then 
$W'^2 = W'(W') \ge W'W \ge WW$, and similarly for higher powers.} Hence
$$
(\psi W')^k \ge (\psi W)^k\quad\text{for all }k\ge 1,
$$
and the same holds for the matrix series. Therefore
$$
\sum_{k=0}^\infty (\psi W')^k\mathbf{1} \;\ge\; \sum_{k=0}^\infty (\psi W)^k\mathbf{1}
$$
entrywise (the term $k=0$ is $\mathbf{1}$ in both sums). If $\beta\ge0$, it follows that
$$
a^\star(W') = \beta\sum_{k\ge 0}(\psi W')^k\mathbf{1} \;\ge\; 
\beta\sum_{k\ge 0}(\psi W)^k\mathbf{1} = a^\star(W).
$$
By linearity of the average, 
$$
\bar a^\star(W') = \frac1n\mathbf{1}^\top a^\star(W') \ge \frac1n\mathbf{1}^\top a^\star(W)
= \bar a^\star(W).
$$

\medskip
\noindent\emph{(ii) Monotonicity in $\mu$ and $q$.}
The matrix series $\sum_{k\ge 0}(\psi W)^k$ does not depend on $\mu$ or $q$, so all dependence is 
through
$$
\beta = \frac{\theta-\lambda\mu q}{c}.
$$
We have 
$$
\frac{\partial \beta}{\partial \mu} = -\frac{\lambda q}{c}\le 0 \text{ and }
\frac{\partial \beta}{\partial q} = -\frac{\lambda \mu}{c}\le 0.
$$
From \eqref{eq:astar-Neumann}, $a^\star$ is linear in $\beta$ with nonnegative coefficients, so each component of $a^\star$ (and thus $\bar a^\star$) is (weakly) decreasing in $\mu$ and in $q$.

\medskip
\noindent\emph{(iii) Monotonicity in $\psi,\theta,c,\lambda$.}
We have
$$
\psi = \frac{\phi}{c} \text{ and } \beta = \frac{\theta-\lambda\mu q}{c}.
$$
First, an increase in $\psi$ (holding $\beta$ fixed) increases each factor $(\psi W)^k$, and thus 
increases the matrix series $\sum_{k\ge 0}(\psi W)^k$ entrywise. Thus, for $\beta\ge 0$, $a^\star$ and $\bar a^\star$ are increasing in $\psi$ and hence in $\phi$. Second, $\beta$ is increasing in $\theta$ and decreasing in $\lambda$. Because $a^\star$ is linear in $\beta$ with nonnegative coefficients, $a^\star$ and $\bar a^\star$ are increasing in $\theta$ and decreasing in $\lambda$. Finally, increasing $c$ reduces both $\beta$ and $\psi$ (since $c$ appears in the denominator of each), and thus reduces $a^\star$ and $\bar a^\star$. Formally, higher $c$ scales down the matrix series through $\psi$ and reduces the baseline term $\beta$. Combining these signs establishes the comparative statics in $\psi,\theta,c,\lambda$. This completes the proof.
\end{proof}

\begin{lemma}[Induced direction of change in $m$ and $b$]\label{lem:mb}
Let $\delta_{LQ} = \delta_0 + \eta \bar a^\star$ with $\eta>0$, and define
$$
m = (1-\delta_{LQ})(1-\mu)\alpha,
\qquad
b = (1-\delta_{LQ})\mu(1-\kappa).
$$
Then:
\begin{enumerate}
\item[(i)] If $W$ strengthens (entrywise), or $\psi$ rises, or $\theta$ rises, then $\bar a^\star$ increases, hence $\delta_{LQ}$ increases and both $m$ and $b$ decrease (holding $\mu,\kappa,\alpha$ fixed).

\item[(ii)] If $\mu$ or $q$ rises, then $\bar a^\star$ falls and hence $\delta_{LQ}$ falls. As a result, $b$ increases. Moreover, $m$ increases whenever the induced decrease in verification is sufficiently strong, i.e. when the increase in $(1-\delta_{LQ})$ dominates the direct attenuation effect from $(1-\mu)$.
\item[(iii)] If $\alpha$ rises (holding $\bar a^\star$ fixed), then $m$ increases proportionally and $b$ is unchanged; if $\kappa$ rises, then $b$ falls proportionally and $m$ is unchanged.
\end{enumerate}
\end{lemma}

\begin{proof}
From Lemma~\ref{lem:MCS}, if $W$ strengthens, or $\psi$ or $\theta$ rises, then $\bar a^\star$ increases. Since $\delta_{LQ} = \delta_0 + \eta\bar a^\star$ with $\eta>0$, this implies $\delta_{LQ}$ increases and $1-\delta_{LQ}$ decreases. Both $m$ and $b$ are proportional to $(1-\delta_{LQ})$, so they decrease.

If $\mu$ or $q$ rises, Lemma~\ref{lem:MCS} implies $\bar a^\star$ decreases, so $\delta_{LQ}$ falls and $1-\delta_{LQ}$ rises. At the same time, $b$ is also proportional to $\mu$. Therefore:
$$
m = (1-\delta_{LQ})(1-\mu)\alpha
\quad\text{and}\quad
b = (1-\delta_{LQ})\mu(1-\kappa)
$$
While $b$ increases unambiguously, the effect on $m$ is in general ambiguous because $\mu$ enters $m=(1-\delta_{LQ})(1-\mu)\alpha$ through two channels: (i) the direct attenuation factor $(1-\mu)$ decreases in $\mu$, while (ii) the verification channel increases $(1-\delta_{LQ})$ because $\bar a^\star$ falls when $\mu$ or $q$ rises. Formally,
$$
\frac{d\log m}{d\mu}=\frac{d\log(1-\delta_{LQ})}{d\mu} -\frac{1}{1-\mu}.
$$
Thus $m$ increases whenever the induced fall in verification is sufficiently strong, i.e. when
$$
\frac{d\log(1-\delta_{LQ})}{d\mu}>\frac{1}{1-\mu}.
$$
This condition is naturally satisfied in regions where verification incentives are highly sensitive to AI exposure (large $\eta$ or steep best-response), including near the stability boundary where small changes in the correction rate materially affect amplification. In contrast, $b$ increases unambiguously because it is proportional to $\mu$. Finally, if $\alpha$ rises holding $\bar a^\star$ fixed, then $\delta_{LQ}$, $\mu$, and $\kappa$ are unchanged, so $m$ increases linearly in $\alpha$ and $b$ is unchanged. If $\kappa$ rises, only the factor $(1-\kappa)$ changes, so $b$ decreases linearly in $(1-\kappa)$ (equivalently, $b$ is affine decreasing in $\kappa$) while $m$ remains unchanged. This proves the lemma.
\end{proof}

\section{Appendix: Derivation Sketch for the Linear Distortion Recursion}\label{app:derivation_contagion}

This appendix provides a derivation sketch motivating the linear recursion \eqref{eq_contagion} as a local (near truth) approximation to the evolution of informational distortion magnitudes. The goal is not to derive an exact identity for KL divergence, but rather to show how a linear law of motion for \emph{distortion levels} can arise from linear belief exchange combined with a local quadratic expansion of KL divergence.

\paragraph{Linear mixing of belief distributions.}
Let $\mathcal{X}$ be a finite space of states and $P_{i,t}\in\mathcal{P}=\Delta(\mathcal{X})$ be agent $i$'s belief
distribution at date $t$. Let $\hat P$ denote the truth distribution. Consider a DeGroot linear mixing of belief distributions with AI input and verification:
\begin{equation}\label{eq:belief_mixing}
P_{i,t+1}=\delta\,\hat P+(1-\delta)\Big[(1-\mu)\sum_{j=1}^n S_{ij} P_{j,t} + \mu\, Q_t\Big],
\end{equation}
where $S$ is a primitive row-stochastic attention matrix ($\sum_j S_{ij}=1$), $Q_t\in\mathcal{P}$ is the AI output distribution at date $t$, $\mu\in[0,1]$ is the AI penetration rate, and $\delta\in(0,1)$ represents verification/correction (a partial reset toward truth).

\paragraph{Local error representation.}
Define belief errors relative to truth by
$$
\varepsilon_{i,t}:=P_{i,t}-\hat P \text{ and } \eta_t:=Q_t-\hat P.
$$
Subtracting $\hat P$ from both sides of \eqref{eq:belief_mixing} yields
\begin{equation}\label{eq:error_mixing}
\varepsilon_{i,t+1}=(1-\delta)\Big[(1-\mu)\sum_{j=1}^n S_{ij}\varepsilon_{j,t} + \mu\, \eta_t\Big].
\end{equation}

\paragraph{Local quadratic expansion of KL divergence.}
Define informational distortions by KL divergence:
$$
e_{i,t}:=D_{\mathrm{KL}}(P_{i,t}\,\|\,\hat P) \text{ and } q_t:=D_{\mathrm{KL}}(Q_t\,\|\,\hat P).
$$
Assume a near-truth regime in which $\|\varepsilon_{i,t}\|$ and $\|\eta_t\|$ are small (in any norm on
$\mathbb{R}^{|\mathcal{X}|}$). A second-order Taylor expansion of KL divergence around $\hat P$ implies
\begin{equation}\label{eq:kl_quadratic}
D_{\mathrm{KL}}(\hat P+\varepsilon\,\|\,\hat P) = \frac{1}{2}\,\varepsilon^\top H(\hat P)\,\varepsilon +o(\|\varepsilon\|^2),
\end{equation}
where $H(\hat P)$ is the Fisher information matrix (the Hessian of KL at $\hat P$). Thus, to second order, KL distortion is approximately a quadratic form in local belief errors.

\paragraph{From linear error mixing to a reduced-form distortion recursion.}

Combining \eqref{eq:error_mixing} with the quadratic approximation \eqref{eq:kl_quadratic} implies that, to second order, distortions take the form of quadratic expressions in the belief errors:
$$
e_{i,t}\approx\frac{1}{2}\,\varepsilon_{i,t}^\top H(\hat P)\,\varepsilon_{i,t}.
$$
Substituting the linear mixing equation \eqref{eq:error_mixing} into this quadratic expression generates terms involving quadratic combinations of neighbors’ errors, including cross-products of the form $\varepsilon_{j,t}^\top H(\hat P)\,\varepsilon_{k,t}$. Passing from this quadratic representation in $\varepsilon_{i,t}$ to a linear recursion in the scalar distortion magnitudes $e_{i,t}$ therefore requires an additional closure assumption. In particular, we adopt a reduced-form first-order approximation under which:
\begin{itemize}
    \item[(i)] cross-agent covariance terms are either negligible or proportional to marginal variances (e.g.\ under a mean-field or weak-correlation regime), and
    \item[(ii)] the scalar distortion index $e_{i,t}$ is proportional to the quadratic form in \eqref{eq:kl_quadratic}.
\end{itemize}

Under such a closure, we adopt a first-order reduced-form approximation in which the scalar distortion index $e_{i,t}$ evolves proportionally to the attention-weighted average of neighbors’ distortion magnitudes and to AI-originating distortion. This amounts to normalizing the quadratic representation in \eqref{eq:kl_quadratic} and absorbing higher-order and covariance terms into proportionality constants. Under these closure assumptions, the resulting dynamics can be expressed as a first-order linear recursion in the distortion magnitudes:
\begin{equation}\label{eq:distortion_recursion_sketch}
e_{t+1}\approx (1-\delta)(1-\mu)\,W\,e_t + (1-\delta)\mu\,\mathbf{1}\,q_t,
\end{equation}
where $e_t=(e_{1,t},\dots,e_{n,t})^\top$ and $W$ is the effective social propagation operator. Equation \eqref{eq:distortion_recursion_sketch} should therefore be interpreted as a first-order reduced-form dynamic for distortion magnitudes induced by linear belief exchange near the truth, rather than as an exact identity for KL divergence. The coefficients $(1-\delta)(1-\mu)$ and $(1-\delta)\mu$ should therefore be interpreted as effective first-order propagation weights that summarize the impact of social transmission and AI injection after the closure approximation. They are reduced-form parameters rather than exact coefficients obtained directly from the quadratic expansion of KL divergence.

\section{Appendix: Network Topology with Finite Horizon Regulation}\label{finitehorizon}

This appendix shows that the homophily and core-periphery results also hold and extend for finite-horizon safety constraints.

\begin{prop}[Homophily increases the stabilizing filter over finite horizons]\label{prop3b}
Fix $(\mu,\delta,\gamma,\alpha)$ and a horizon $T\in\mathbb{N}$, and let $x_t=(\bar e_t,q_t)^\top$ denote the aggregate state. Define $\kappa_T^\star(\mu;S)$ as the minimal $\kappa\in[0,1]$ such that the envelope bound implied by Lemma~\ref{lemma1} ensures
$$
\max_{0\le t\le T}\|x_t\|\;\le\;\bar X,
$$
for a fixed tolerance $\bar X>0$ and initial condition $x_0$. If $S$ and $S'$ are primitive row-stochastic and satisfy $|\lambda_2(S')|>|\lambda_2(S)|$, then
$$
\kappa_T^\star(\mu;S') \;\ge\; \kappa_T^\star(\mu;S).
$$
\end{prop}

The asymptotic stability frontier $\kappa^\star(\mu)$ derived in Section~\ref{2dim} characterizes long-run stability of the reduced two-dimensional system and depends only on the dominant amplification parameter $\alpha$. By contrast, Proposition~\ref{prop3b} captures finite-horizon stabilization requirements driven by the persistence of subdominant network modes. Homophily, as measured by $|\lambda_2(S)|$, slows the decay of local distortions and therefore raises the minimal filtering intensity required to keep the AI-society system within safe bounds over economically relevant horizons, even when asymptotic stability is unchanged.

\begin{proof}
Fix $(\mu,\delta,\gamma,\alpha)$, a horizon $T\in\mathbb{N}$, and $\kappa\in[0,1]$.
Now fix two primitive row-stochastic networks $S$ and $S'$ such that $|\lambda_2(S')|>|\lambda_2(S)|$, and define
$$\rho_S:=\frac{1+|\lambda_2(S)|}{2} \text{ and }
\rho_{S'}:=\frac{1+|\lambda_2(S')|}{2}.$$
Then $\rho_S\in(|\lambda_2(S)|,1)$, $\rho_{S'}\in(|\lambda_2(S')|,1)$, and $\rho_{S'}>\rho_S$. By Lemma~\ref{lemma1}, the aggregate distortion satisfies
$$\bar e_{t+1}=m\bar e_t+bq_t+\varepsilon_t(S), m=(1-\delta)(1-\mu)\alpha \text{ and } b=(1-\delta)\mu(1-\kappa),$$
with $|\varepsilon_t(S)|\le C(S)(m\rho_S)^t$ for all $t\ge 0$ and similarly $|\varepsilon_t(S')|\le C(S')(m\rho_{S'})^t$ for all $t\ge 0$. Define a \emph{common} constant
$$\bar C:=\max\{C(S),C(S')\}<\infty.$$
Hence, for all $t\ge0$, $|\varepsilon_t(S)|\le \bar C(m\rho_S)^t$ and $|\varepsilon_t(S')|\le \bar C(m\rho_{S'})^t$. We use this shared constant $\bar C$ in the envelope comparison below. Iterating forward gives, for each $t\le T$,
$$\bar e_t=m^t\bar e_0 + \sum_{s=0}^{t-1}m^{t-1-s}bq_s + \sum_{s=0}^{t-1}m^{t-1-s}\varepsilon_s(S).$$

Taking absolute values and using the common bound on $\varepsilon_s(S)$ yields to
$$
|\bar e_t|\le m^t|\bar e_0| + \sum_{s=0}^{t-1}m^{t-1-s}b|q_s| + \bar C\,m^{t-1}\sum_{s=0}^{t-1}\rho_S^s.$$

The corresponding envelope under $S'$ is obtained by replacing $\rho_S$ with $\rho_{S'}$. Since $\rho_{S'}>\rho_S$, the last term under $S'$ is weakly larger than under $S$ for every $t\le T$.

Since $q_{t+1}=\gamma q_t+(1-\gamma)\bar e_t$, the common-envelope construction shows that the persistence term
$\bar C\,m^{t-1}\sum_{s=0}^{t-1}\rho_S^s$ appearing in the envelope bound is increasing in $\rho_S$. Since $\rho_{S'}>\rho_S$, the sufficient bound used to ensure $\max_{0\le t\le T}\|x_t\|\le\bar X$ is weakly harder to satisfy under $S'$ than under $S$. Therefore, the minimal filtering level defined by this envelope criterion satisfies $\kappa_T^\star(\mu;S')\ge\kappa_T^\star(\mu;S)$.
\end{proof}

\paragraph{Targeted regulation in centralized networks.}
The previous analysis shows that in core-periphery networks systemic informational risk concentrates disproportionately on a small set of high-centrality nodes. This raises a natural policy question: if the regulator can apply upstream filtering unevenly across agents or platforms, should filtering effort be concentrated on the most influential nodes? In practice, regulatory obligations are rarely uniform. Large platforms, recommender systems, and highly visible accounts are typically subject to stricter auditing, alignment requirements, and content controls than ordinary users. We formalize this asymmetry by allowing node-specific upstream filtering of AI-originating distortions.

\medskip

Rather than imposing a uniform filtering intensity, the regulator may choose node-specific filtering levels $\kappa_i \in [0,1]$, where $\kappa_i$ represents the fraction of AI-generated distortion removed before reaching agent $i$.

Let $\kappa = (\kappa_1,\dots,\kappa_n)$ denote the vector of filtering intensities. Regulatory capacity is limited by an average budget constraint. We represent this as a fixed filtering budget, so that the regulator chooses $\kappa \in [0,1]^n$ subject to
$$\mathcal K(\bar\kappa)
=
\left\{
\kappa\in[0,1]^n:
\frac1n\sum_{i=1}^n \kappa_i=\bar\kappa
\right\},
\qquad \bar\kappa \in [0,1].$$

Uniform filtering corresponds to the allocation $\kappa_i \equiv \bar\kappa$ for all $i$, which provides a natural reference point. To evaluate regulatory performance over an economically relevant horizon, we consider the discounted cumulative aggregate distortion

$$
\mathcal{L}_T(\kappa):=\sum_{t=0}^{T} \beta^t \,\bar e_t,\qquad \beta \in (0,1),
$$
where $\bar e_t$ denotes aggregate distortion. In this subsection, aggregate distortion is defined using the stationary distribution of the attention matrix,

$$
\bar e_t := \pi^\top e_t,
$$
where $\pi$ is the unique stationary distribution of the primitive row-stochastic matrix $S$. This choice captures long-run influence weights in the network and is particularly appropriate in core-periphery structures, where informational influence is concentrated on a subset of nodes. Importantly, under this $\pi$-weighted aggregation the aggregate system closes exactly under the full AI feedback loop. Since $W=\alpha S$ and $\pi^\top S=\pi^\top$, it follows that
$$
\pi^\top M = m\,\pi^\top, \;\; M := (1-\delta)(1-\mu)W, \text{ and } m := (1-\delta)(1-\mu)\alpha.
$$
Consequently, even with node-specific filtering, the joint dynamics of $(\bar e_t,q_t)$ form an exact two-dimensional linear system.

\medskip

We compare two classes of policies. Uniform filtering is defined by $\kappa^U_i \equiv \bar\kappa$. A core-targeted policy is any $\kappa^C \in \mathcal{K}(\bar\kappa)$ that assigns weakly higher filtering intensity to nodes with larger stationary influence weights $\pi_i$. For example, ordering nodes such that $\pi_{(1)} \ge \pi_{(2)} \ge \dots \ge \pi_{(n)}$, define
$$
\kappa^C_{(i)} =
\begin{cases}
1 & \text{for } i=1,\dots,K-1, \\
n\bar\kappa-(K-1) & \text{for } i=K, \\
0 & \text{for } i=K+1,\dots,n,
\end{cases}
$$
where $K$ is the unique integer satisfying $K-1 \le n\bar\kappa < K$. This policy exhausts the regulatory budget and allocates filtering on the most influential nodes.

\begin{prop}[Targeted high-centrality filtering under finite-horizon constraints]

\label{prop4b}
Fix $(\mu,\delta,\gamma,\alpha)$ and a finite horizon $T$.
Suppose the regulator chooses $\kappa \in \mathcal{K}(\bar\kappa)$ to minimize $\mathcal{L}_T(\kappa)$ under the full AI feedback loop. If the attention matrix exhibits a core-periphery structure such that the stationary distribution $\pi$ places strictly larger mass on nodes in the core than on peripheral nodes, then there exists an optimal policy $\kappa^\star$ satisfying
$$
\pi_i > \pi_j
\quad \Rightarrow \quad
\kappa_i^\star \ge \kappa_j^\star.
$$
In particular, when $\pi$ concentrates on the core, an optimal policy can be chosen to allocate weakly more filtering to nodes in the core than to peripheral nodes.

Moreover, whenever $\pi$ is non-uniform, a high-centrality targeted policy weakly dominates uniform filtering, with strict dominance whenever the stationary weights differ across nodes. 
\end{prop}

\begin{proof}
The proof proceeds in three steps. First, we show that under $\pi$-weighted aggregation the full AI-society feedback system closes exactly as a two-dimensional linear recursion, and that the dynamics depend on the filtering vector $\kappa$ only through the scalar $\theta(\kappa)=1-\pi^\top\kappa$. Second, we establish that aggregate distortion over any finite horizon is (weakly) increasing in $\theta(\kappa)$. Third, we show that minimizing the finite-horizon loss is therefore equivalent to maximizing the linear functional $\pi^\top\kappa$ over the feasible set, which implies that optimal regulation assigns weakly higher filtering intensity to nodes with larger stationary influence weights.

\medskip 

We incorporate the full feedback loop between society and AI. Under node-specific filtering $\kappa\in[0,1]^n$, the distortion dynamics are
\begin{align}
e_{t+1} &= (1-\delta)(1-\mu)W e_t + (1-\delta)\mu\, q_t\,(\mathbf{1}-\kappa), \label{eq:e_dynamics}\\
q_{t+1} &= \gamma q_t + (1-\gamma)\bar e_t \text{ and }\bar e_t := \pi^\top e_t,
\label{eq:q_dynamics}
\end{align}
where $\pi$ is its unique stationary distribution satisfying $\pi^\top S=\pi^\top$ and $\pi^\top \mathbf{1}=1$.

\medskip

\paragraph{Step (i): Exact aggregate closure under feedback.}

Define
$$
M := (1-\delta)(1-\mu)W = (1-\delta)(1-\mu)\alpha S =: m S \text{ and } m := (1-\delta)(1-\mu)\alpha.
$$

Because $\pi^\top S=\pi^\top$, we obtain $\pi^\top M = m\,\pi^\top$. Premultiplying \eqref{eq:e_dynamics} by $\pi^\top$ yields
$$
\bar e_{t+1}= m\,\bar e_t + (1-\delta)\mu\,q_t\,\pi^\top(\mathbf{1}-\kappa).
$$

Define now $a := (1-\delta)\mu$ and $\theta(\kappa):= \pi^\top(\mathbf{1}-\kappa)=1-\pi^\top\kappa$. Then the aggregate dynamics reduce exactly to
\begin{equation}\label{eq:agg_e}\\
\bar e_{t+1} = m\,\bar e_t + a\,\theta(\kappa)\,q_t \text{ and } q_{t+1} = \gamma q_t + (1-\gamma)\bar e_t.
\end{equation}

Equations \eqref{eq:agg_e}  form a closed
two-dimensional linear system: 
$$x_{t+1}= A(\kappa)\,x_t, \text{ and } x_t := \begin{pmatrix}
\bar e_t \\
q_t
\end{pmatrix} \text{ with } A(\kappa) =
\begin{pmatrix}
m & a\,\theta(\kappa) \\
1-\gamma & \gamma
\end{pmatrix}.$$

Hence the entire trajectory $(\bar e_t,q_t)_{t=0}^T$ depends on $\kappa$ only through the scalar $\theta(\kappa)=1-\pi^\top\kappa$.

\paragraph{Step (ii): Monotonicity of distortions in $\theta$.}

Let initial conditions satisfy $\bar e_0 \ge 0$ and $q_0 \ge 0$. All entries of $A(\kappa)$ are nonnegative for $\kappa\in[0,1]^n$. Therefore, by induction, $\bar e_t \ge 0$ and $q_t \ge 0$ for all $t$. Moreover, if $\theta_1 > \theta_2$, then
$A(\kappa_1) \ge A(\kappa_2)$ componentwise, with strict inequality in the $(1,2)$ entry. By monotonicity of linear systems with nonnegative matrices,
$$
\theta_1 > \theta_2
\quad \Rightarrow \quad
\bar e_t(\theta_1) \ge \bar e_t(\theta_2)
\quad \text{for all } t.
$$

Hence the finite-horizon loss $\mathcal L_T(\kappa) = \sum_{t=0}^T \beta^t \bar e_t$ is weakly increasing in $\theta(\kappa)$, and strictly increasing whenever $q_t$ is strictly positive for some $t<T$. Because $\theta(\kappa)=1-\pi^\top\kappa$, minimizing $\mathcal L_T(\kappa)$ is equivalent to maximizing $\pi^\top\kappa$ over the feasible set $\mathcal K(\bar\kappa)=\left\{\kappa\in[0,1]^n:\frac1n\sum_{i=1}^n\kappa_i= \bar\kappa\right\}$.

\paragraph{Step (iii): Optimal allocation of filtering effort.}

The problem reduces to
$$
\max_{\kappa\in\mathcal K(\bar\kappa)} \ \pi^\top\kappa = \sum_{i=1}^n \pi_i \kappa_i.
$$

This is a linear program with linear objective and convex polyhedral feasible set. Let $\kappa^\star$ be an optimal solution. Suppose there exist indices $i,j$ such that $\pi_i > \pi_j$ but $\kappa_i^\star < \kappa_j^\star$. Define now, for sufficiently small $\Delta>0$,
$$
\tilde\kappa_i = \kappa_i^\star + \Delta, \;\; \tilde\kappa_j = \kappa_j^\star - \Delta \text{ and } \tilde\kappa_\ell = \kappa_\ell^\star\ \text{for } \ell\neq i,j,
$$
with $\Delta$ chosen so that $\tilde\kappa\in\mathcal K(\bar\kappa)$. Feasibility is preserved since the average constraint depends only on the sum of components. Then $\pi^\top \tilde\kappa - \pi^\top \kappa^\star = \Delta(\pi_i-\pi_j) > 0$, contradicting optimality of $\kappa^\star$. Therefore any optimizer must satisfy $$\pi_i > \pi_j \Rightarrow \kappa_i^\star \ge \kappa_j^\star.$$

Note that in a core-periphery network, the stationary distribution $\pi$ assigns strictly larger weights to core nodes than to peripheral nodes. Hence an optimal policy can be chosen so that nodes in the core receive weakly greater filtering intensity than peripheral nodes.

Uniform filtering $\kappa_i^U\equiv\bar\kappa$ is feasible. If $\pi$ is non-uniform, uniform filtering does not maximize $\pi^\top\kappa$ over $\mathcal K(\bar\kappa)$ whenever the constraint set admits a non-uniform feasible reallocation. Therefore, there exists a core-targeted policy $\kappa^C$ such that
$$\mathcal L_T(\kappa^C)\le \mathcal L_T(\kappa^U),$$
with strict inequality whenever the stationary weights differ and the constraint set admits a non-uniform feasible reallocation.
\end{proof}

\section{Appendix: Information Projection to Regulatory Filtering}\label{app:Iproj}

This appendix explains the information-theoretic projection (I-projection) used to model regulatory filtering of distortions originating from AI. We first recall the geometry of KL divergence and I-projection, following \cite{Csiszar}, and then show how this tool enters the AI-society dynamics through the injection parameter $b$ and the matrix $A(\mu,\kappa)$.

\paragraph{Why KL divergence and not another distance.}

Several distances or divergences could in principle be used to measure the discrepancy between belief distributions (e.g.\ total variation, $\ell_2$ distance, Wasserstein metrics). We adopt Kullback-Leibler divergence for three reasons.

First, KL divergence is the natural loss function underlying Bayesian updating and likelihood-based
learning, which closely mirrors how AI systems are trained and evaluated. Using KL therefore aligns the notion of informational distortion in the model with the statistical and algorithmic foundations of AI systems. Second, KL divergence admits a well-defined information projection (I-projection) onto convex constraint sets, with existence, uniqueness, and the Pythagorean inequality \citep{Csiszar}. These properties make KL uniquely suited for modeling regulatory filtering as a minimal, disciplined correction of AI outputs subject to truth constraints. Few alternative distances admit an analogous projection operator with comparable geometric and optimization properties. Third, while KL divergence provides the micro-level interpretation of distortion and correction, our aggregate results do not rely on the fine geometry of the divergence. Once linearized, regulation enters the dynamics only through the scalar parameter $\kappa$, which summarizes the amount of AI-origin distortion removed. In this sense, KL divergence is not restrictive: any divergence that supports a well-defined, monotone notion of informational correction would lead to the same reduced form dynamics. KL is chosen because it provides a canonical and information-theoretically grounded foundation for that correction.

\paragraph{KL divergence and I-projection.}

Let $\mathcal{X}$ be a finite set of states, and let $\mathcal{P}$ denote the simplex of probability distributions on $\mathcal{X}$. For $P,Q\in\mathcal{P}$ with $P$ absolutely continuous with respect to $Q$, the Kullback-Leibler (KL) divergence is
$$
D_{\mathrm{KL}}(P\|Q)
:= \sum_{x\in\mathcal{X}} P(x)\log\frac{P(x)}{Q(x)}.
$$
This is a nonnegative function, $D_{\mathrm{KL}}(P\|Q)\ge 0$ (by Gibbs' inequality), and $D_{\mathrm{KL}}(P\|Q)=0$ if and only if $P=Q$. It is not a symmetric distance, but it is the fundamental divergence measure in information theory and Bayesian updating.

Let $\mathcal{C}\subseteq\mathcal{P}$ be a nonempty, convex set of distributions, representing a set of constraints (in our application, a set of truth-consistent or admissible distributions). Given an ``unconstrained'' distribution $Q\in\mathcal{P}$ (the AI's current output), the I-projection of $Q$ onto $\mathcal{C}$ is defined as
\begin{equation}\label{eq:Iproj-def}
P^\star := \arg\min_{P\in\mathcal{C}} D_{\mathrm{KL}}(P\|Q).
\end{equation}
Intuitively, $P^\star$ is the closest distribution to $Q$ that satisfies the constraints defining $\mathcal{C}$, where closeness is measured in KL divergence. \cite{Csiszar} shows that under mild conditions (e.g. $\mathcal{C}$ is convex and closed in the KL topology), the I-projection exists, is unique, and admits a useful geometric interpretation.

A fundamental property of the I-projection onto a closed convex set is the Pythagorean inequality \citep[see, e.g.,][]{Csiszar}. Let $\mathcal{C}$ be a nonempty convex set and let $P^{\star}$ be the I-projection of $Q$ onto $\mathcal{C}$. Then for any $P \in \mathcal{C}$,
\begin{equation}\label{eq:pythagoras}
D_{\mathrm{KL}}(P \| Q) \ge D_{\mathrm{KL}}\left(P \| P^{\star}\right)
+ D_{\mathrm{KL}}\left(P^{\star} \| Q\right).
\end{equation}

Equality holds for all $P\in\mathcal{C}$ when $\mathcal{C}$ is an affine family. In general only the inequality is guaranteed. In particular,
$$
D_{\mathrm{KL}}(P\|Q) \ge D_{\mathrm{KL}}(P^\star\|Q),
$$
with equality if and only if $P=P^\star$. Equation \eqref{eq:pythagoras} states that the I-projection $P^\star$ is the unique point in $\mathcal{C}$ that minimizes KL divergence to $Q$, and that any other $P\in\mathcal{C}$ is further from $Q$ in KL both directly and via $P^\star$.

In our context, $\mathcal{C}$ will play the role of a truth-consistent set of distributions, and the I-projection selects the minimal KL correction required to make $Q$ compatible with these constraints.

\paragraph{Interpretation in the AI-society model.}

In the model, the AI at date $t$ produces a distribution $Q_t$ over states of knowledge (reflecting its internal representation of the world or the distribution of its answers), which may be biased or distorted relative to the truth distribution $\hat P$. Let $\mathcal{\hat P}$ denote a set of truth-consistent or admissible distributions (for instance, distributions satisfying known constraints, calibration properties, or externally verified facts).

The regulator applies the I-projection to $Q_t$:
$$
P^\star_t := \arg\min_{P\in\mathcal{\hat P}} D_{\mathrm{KL}}(P\|Q_t).
$$
This operator has three crucial features:

\begin{enumerate}
\item It produces a \emph{truth-consistent} distribution $P^\star_t\in\mathcal{\hat P}$.
\item It is \emph{KL-optimal}: among all such truth-consistent distributions, $P^\star_t$ is the closest to $Q_t$ in KL divergence. It therefore preserves as much of the AI's informational content as possible, subject to truth constraints.
\item The minimality property $D_{\mathrm{KL}}(P\|Q) \ge D_{\mathrm{KL}}(P^\star\|Q)$ follows directly from the definition of $P^\star$ as the argmin of $D_{\mathrm{KL}}(P\|Q)$ over $\mathcal C$. The Pythagorean inequality \eqref{eq:pythagoras} provides the geometric decomposition
$$ D_{\mathrm{KL}}(P \| Q) \ge D_{\mathrm{KL}}(P \| P^\star) + D_{\mathrm{KL}}(P^\star \| Q),$$

which ensures uniqueness under convexity of $\mathcal C$ and strict convexity of the KL divergence in its first argument, and interprets $P^\star$ as the orthogonal projection of $Q$ onto $\mathcal C$ in KL geometry.
\end{enumerate}

Economically, $Q_t$ is the AI's unfiltered output, and $P^\star_t$ is the filtered version that the regulator permits to enter the social information environment. The difference between $Q_t$ and $P^\star_t$ measures how much distortion is removed by regulation. We summarize the \emph{overall intensity} of this filtering with a scalar parameter $\kappa\in[0,1]$: $\kappa$ measures the fraction of informational distortion originating from AI removed by I-projection, while $(1-\kappa)$ measures the fraction that survives. Formally, we interpret the scalar $q_t$ as a reduced-form measure of AI distortion defined in KL units, for example $q_t := D_{\mathrm{KL}}(Q_t \,\|\, \hat P)$, where $\hat P$ denotes the ground-truth distribution. The precise normalization of $q_t$ is immaterial for the linear dynamics, which depend only on its magnitude up to scale.

\paragraph{How I-projection enters the dynamics.}
It is important to note that, at the distributional level, the I-projection defined in \eqref{eq:Iproj-def} generally produces a nonlinear and state-dependent correction: the ratio between pre and post-projection KL divergences does not need to be constant across $Q_t$. However, the aggregate system studied in the main text is derived as a local, near-truth linear approximation. In this reduced-form representation, the effect of upstream filtering on the scalar distortion index is summarized by a constant parameter $\kappa \in [0,1]$, which captures the average or local proportional removal rate of AI-origin distortion. Thus, while the underlying I-projection operates at the distribution level, its impact on the linearized macro dynamics is represented by the multiplicative factor $(1-\kappa)$. In the local linear approximation considered throughout the paper, we treat this proportional removal rate as constant and summarize it by a parameter $\kappa\in[0,1]$. After filtering, the social contagion equation \eqref{eq_contagion} takes the form

$$
e_{t+1} = (1-\delta)(1-\mu) W e_t + (1-\delta)\mu(1-\kappa)q_t\mathbf{1}.
$$
Relative to an unregulated case, in which the term multiplying $q_t$ would be $(1-\delta)\mu$, the I-projection reduces distortion originating from AI by the factor $(1-\kappa)$. In the notation of the main text, this determines the injection parameter 
\begin{equation*}\label{eq:b-Iproj}
b = (1-\delta)\mu(1-\kappa).
\end{equation*}

The reduced two-dimensional dynamics, derived in Section~\ref{2dim}, are
$$
x_{t+1}=A(\mu,\kappa)x_t,\qquad 
A(\mu,\kappa)=\begin{bmatrix} m & b\\ 1-\gamma & \gamma\end{bmatrix},
\quad
x_t=\begin{bmatrix}\bar e_t\\ q_t\end{bmatrix},
$$
where $m=(1-\delta)(1-\mu)\alpha$ represents internal amplification through social interaction intensity, and $b$ measures the effective injection originating from AI after regulatory filtering.

\paragraph{Role of I-projection in the main results.}
Because the eigenvalues of $A(\mu,\kappa)$ determine whether distortions converge or explode, the regulator's ability to influence $b$ via $\kappa$ is what makes regulation effective. In particular, the systemic-contagion boundary is characterized by the condition $\rho(A(\mu,\kappa))=1$, which is equivalent in our $2\times 2$ setting to
$$
b = 1 - m.
$$
Substituting $m=(1-\delta)(1-\mu)\alpha$ and $b=(1-\delta)\mu(1-\kappa)$ from \eqref{eq:b-Iproj} yields the policy frontier
$$
\kappa^\star(\mu) = 1 - \frac{1-\alpha(1-\delta)(1-\mu)}{(1-\delta)\mu}.
$$
Thus, the I-projection underlies the regulator's control of the injection parameter $b$ and, through it, of the entire stability region. Without the I-projection, KL divergence would still provide a way to measure informational distortion, but there would be no disciplined operator to reduce it. The I-projection fills precisely this role by reducing a high-dimensional distributional correction problem to a scalar policy instrument $\kappa$ that enters linearly in the amplification matrix $A(\mu,\kappa)$ and therefore in all stability and policy results.

\bibliography{biblio}
\end{document}